\documentclass[aps,twocolumn,superscriptaddress,showpacs,showkeys]{revtex4}
\usepackage{setspace,array,soul,amsthm}
\usepackage{graphics,graphicx,dcolumn,bm,fleqn,epic,eepic,float}
\usepackage{amssymb,amsmath,multirow,rotate,color,epstopdf,bbm,times}
\definecolor{red}{rgb}{1,0,0}
\definecolor{green}{rgb}{0,1,0}
\definecolor{blue}{rgb}{0,0,1}
\newtheorem{proposition}{Proposition}[section]

\DeclareMathOperator{\Tr}{Tr}
\DeclareMathOperator{\Det}{Det}
\DeclareMathOperator{\sgn}{sgn}
\begin{document}
%%%%%%%%%%%%%%%%%%%%%%%%%%%%%%%%%%%%%%%%%%%%%%%%%%%%%%%%%%%%%%%%%%%%%%%%%%

%\title{Exploring the existence of continuous Markov processes in
%       empirical data: a systematic approach}
%\title{The time-inhomogeneous embedding problem: 
%         a general analytical and numerical approach}
\title{From empirical data to continuous Markov processes: a systematic
       approach}

\author{Pedro Lencastre}
\affiliation{Mathematical Department, FCUL, University of Lisbon, 1749-016 Lisbon, Portugal}
\author{Frank Raischel}
\affiliation{Instituto Dom Luiz, University of Lisbon, 1749-016 Lisbon, Portugal}
\author{Tim Rogers}
\affiliation{Centre for Networks and Collective Behaviour, Department of Mathematical Sciences, University of Bath, Claverton Down, BA2 7AY, Bath, UK}
\author{Pedro G.~Lind}
%\affiliation{TWIST - Turbulence, Wind energy and Stochastics, 
%            Institute of Physics, Carl-von-Ossietzky University 
%            of Oldenburg, DE-26111 Oldenburg, Germany}
\affiliation{ForWind - Center for Wind Energy Research, Institute of Physics,
             Carl-von-Ossietzky University of Oldenburg, DE-26111 Oldenburg, 
             Germany}
\affiliation{%
         Institut f\"ur Physik, Universit\"at Osnabr\"uck,
         Barbarastrasse 7, 49076 Osnabr\"uck, Germany}
\date{\today}

\begin{abstract}
We present an approach for testing for the existence of continuous
generators of discrete stochastic transition matrices.
Typically, the known approaches to ascertain the existence of continuous 
Markov processes are based in the assumption that only time-homogeneous
generators exist. Here, a systematic extension to time-inhomogeneity 
is presented, based in new mathematical propositions incorporating 
necessary and sufficient conditions, which are then implemented 
computationally and applied to numerical data.
A discussion concerning the bridging between rigorous 
mathematical results on the existence of generators to its
computational implementation. 
Our detection algorithm shows to be effective in more than 
$80\%$ of tested matrices, 
typically $90\%$ to $95\%$, and for those an 
estimate of the (non-homogeneous) generator matrix follows.
We also solve the embedding problem analytically for the particular case of
three-dimensional circulant matrices.
Finally, a discussion of possible applications of our framework to 
problems in different fields is briefly addressed.
\end{abstract}

%%%%PACS e Keywords
\pacs{02.50.Ga,  %Markov processes
      %02.50.Ey, %Stochastic processes
      05.10.Gg,  %Stochastic processes in 
                 %statistical physics and nonlinear dynamics, 
      02.10.Yn,  %Matrix theory
      89.65.Gh}  %Financial markets

\keywords{Continuous Markov Processes, Embedding Problem,
Inhomogeneous Generators, Master Equation}  %Rating Agencies}

\maketitle

%\tableofcontents

%%%%%%%%%%%%%%%%%%%%%%%%%%%%%%%%%%%%%%%%%%%%%%%%%%%%%%%%%%%%%%%%%%%%%%%
\section{Motivation}

While models describing the evolution of a set of variables
are typically continuous, observations and experiments retrieve
discrete sets of values. Therefore, to bridge between models and 
reality one has to know if it is reasonable to assume a continuous 
``reality'' underlying the discrete set of measurements.
When the evolution has a non negligible stochastic contribution,
one typically extracts from the set of measurements the distribution
$\vec{P}(X,t-\tau)$ of the observed values $X(t-\tau)$, from which the 
probability density function (PDF) can be inferred. By knowing the 
distribution
$\vec{P}(X,t)$ at a future time $t$, one is then able to
define a transition matrix $\mathbf{T}(t,\tau)$ that satisfies:
\begin{equation} 
\vec{P}(X,t) = \mathbf{T}(t,\tau)\vec{P}(X,t-\tau) ,
\label{transeq}
\end{equation} 
if we know the fraction of transitions $T_{kj}$ from each observed value 
$X_j(t-\tau)$ at time $t-\tau$ to a value $X_k(t)$ at time $t$.
The transition matrix $\mathbf{T}(t,\tau)$ has all its elements
$T_{kj}$ in the interval $[0,1]$, has row-sums one, $\sum_{k} T_{kj}=1$, 
and has non-negative entries, $T_{kj}\ge 0$. 

In this paper we address the problem of determining whether or not the 
evolution of a system is governed by a time-continuous Markov master 
equation. 
This problem is usually called the embedding problem~\cite{Elfving1937}. 
Time-continuous Markov processes, have particular mathematical
properties, namely they memoryless
stochastic processes: the probability of transition between states $X(t)$
and $X(t+\tau)$ does not depend on the states of the system for times 
previous to $t$, for any $\tau>0$.
If the stochastic process is time-continuous and Markovian, 
than the transition matrix can be defined for infinitely small $\tau$,
obeying an equation of the form
\begin{equation}
\frac{d\mathbf{T}(t,\tau)}{d\tau}=\mathbf{Q}(t)\mathbf{T}(t,\tau)\, ,
\label{MC1}
\end{equation}
where $\mathbf{Q}(t)$ is called the generator matrix of the process, 
having zero row-sums and non-negative off-diagonal entries. 
Notice that, $\mathbf{T}(t,\tau)$ is a transition matrix for all $t$
and $\tau$, i.e.~with non-negative real elements and unity row-sums,
if and only if it obeys Eq.~(\ref{MC1}) for some 
$\mathbf{Q}(t)$\cite{Kingman1962}.

Both equations above allow us to write the continuous-time evolution of 
a PDF.
In other words, the time evolution of such PDF can be described by a master 
equation in continuous time. 
The Master Equation Approach is a fundamental tool in Statistical Physics
used to derive important results in Thermodynamics\cite{esposito2010three} 
and in several interdisciplinary 
applications\cite{privman2005nonequilibrium}.

The transition matrix $T(t,\tau)$, solution of Eq.~(\ref{MC1}),
defines the evolution equation, Eq.~(\ref{transeq}), of the PDF.
Thus, the entries $Q_{kj}$ of the generator matrix represent the 
transition rate between states $j$ and $k$ at time $t$.
Time-continuity is a property that results from the fact that
all entries of $\mathbf{Q}(t)\equiv \mathbf{Q}$, 
i.e.~all transition rates, are finite.
The general solution of Eq.~(\ref{MC1}) yields the relation between
the empirical transition matrix and the ``continuous'' generator 
which, in the particular case of a time-homogeneous transition matrix, 
has the form
\begin{equation}
\mathbf{T}(t,\tau) \equiv \mathbf{T}(\tau) = \exp{\left (
                          \mathbf{Q}\tau
                          \right )} ,
\label{homogeneous_solution}
\end{equation} 
for all times $t$.
In general, 
the embedding problem reduces to the problem of being able to write
the transition matrix $\mathbf{T}(t,\tau)$ as solution of Eq.~(\ref{MC1})
and typically one considers the particular case of a time-homogeneous 
solution, Eq.~(\ref{homogeneous_solution}).

While time-homogeneity is a useful common assumption, 
it is in several cases too restrictive. 
Assuming time-homogeneity has the advantage of knowing all 
future evolution of a time-homogeneous Markov process from the 
law of the change of system's configuration 
in two distinct instants (see Eq.~(\ref{homogeneous_solution})),
but simultaneously one is not able to address more realistic cases of 
non-stationary systems. 
Previous progress in this topic has been made recently.
Shintani and Shinomoto have examined an optimized Bayesian rate estimator
in cases where the probability density function is not constant in 
time\cite{shintani2012}.
%Urdapilleta presented in Ref.~\cite{urdapilleta2011} a first-order 
%statistical description of an exponentially driven time-inhomogeneous 
%stochastic process for characterizing the correlations present
%in interspike interval series.

In this scope, there are three main reasons for considering an
empirical transition matrix to not be time-homogeneous embeddable.
The first one is when the underlying process is not Markovian.
Such scenario was previously addressed by us\cite{lencastre1,lencastre2}.
One second reason is the statistical error any empirical data set is 
subjected to. 
Typically, one defines for these cases an interval of confidence 
(a distance) beyond which embeddability is rejected.
The third reason is, of course, that the underlying process is itself
not time-homogeneous. In this case, there is no time-homogeneous
generator, but there is sill the chance that an inhomogeneous generator
exists.

In this paper, we address analytically and numerically the case of 
time-inhomogeneous generators and test their implementation in one 
framework to address synthetic numerical data, dealing with statistical
error of transition matrices.
We will also review the time-homogeneous embedding problem, introduced 
in 1937 by Elfving\cite{Elfving1937}, providing an analytical example 
in three-dimensions.

We start in Sec.~\ref{sec:embedding_soa} by describing the standard 
time-homogeneous problem with the main mathematical theorems that give the 
necessary and sufficient conditions for a generator matrix to exist.
In Sec.~\ref{sec:example} we illustrate this standard time-homogeneous
embedding problem by applying the results to the specific case of a 
circulant transition matrix.
Sections \ref{sec:inhomogeneous} and \ref{sec:implementation} are
the heart of this paper, the former establishes the main mathematical 
theorems that are still valid for the general case of inhomogeneous
generators and the latter describes their implementation in a
framework that is then tested with synthetic data.
Finally, discussions and conclusions are given in 
Sec.~\ref{sec:conclusions}.

%%%%
\section{The homogeneous embedding problem}
\label{sec:embedding_soa}

The question of knowing if a time-homogeneous generator $\mathbf{Q}$ 
(see Eq.~(\ref{homogeneous_solution})) exists is known as homogeneous 
embedding problem\cite{Elfving1937} and, from the mathematical point 
of view is 
currently an open problem for matrices with dimension $n\geq 3$. 
The problem in dimension two was solved in 1962 by Kingman\cite{Kingman1962}, 
who proved that, for $n=2$, a matrix is embeddable if and only 
if its determinant is positive. 
More recently, developments in three dimensions were done with the study of 
matrices with repeated negative eigenvalues\cite{3x3generator}. 

Part of the difficulty when addressing the embedding problem
arises from the fact that the logarithm of a matrix 
is, in general, not unique. This is crucial when deriving a generator
$\mathbf{Q}$, by inverting Eq.~(\ref{homogeneous_solution}).
%$\mathbf{Q}=(\hbox{"log"}(\mathbf{T}))/\tau$
Indeed, the logarithm of a matrix has counter-intuitive 
properties, namely: 
\begin{itemize}
\item[(i)] The product of two embeddable transition
matrices $ \mathbf{T}_1$ and $ \mathbf{T}_2$ is also a transition
matrix not necessarily embeddable.
\item[(ii)] Having two transition embeddable matrices with
  generators $\mathbf{Q}_1$ and $ \mathbf{Q}_2$, if their product is embeddable 
  then its generator is not necessarily $ \mathbf{Q}_1+
  \mathbf{Q}_2$, unless the transition matrices commute.  
\item[(iii)] It is possible that the product of two matrices, 
$\mathbf{T}_1 \mathbf{T}_2$, is embeddable, but the product 
$\mathbf{T}_2 \mathbf{T}_1$ is not.
\end{itemize}

Since the logarithm of a matrix is not unique, one defines the so-called
principal logarithm of one matrix $\mathbf{T}$ as\cite{Davies2010}
\begin{equation}
\log \mathbf{T} = \frac{1}{2 \pi i} \int_{\gamma} 
        \log z \left ( z\mathbf{I} - \mathbf{T} \right )^{-1} dz\, ,
\label{principallog}
\end{equation}
where $\gamma$ is a path in the complex plane which does not intersect
the negative real semi-axis and encloses all eigenvalues of 
$\mathbf{T}$. 
Computationally, one uses the Taylor expansion of the logarithmic function,
yielding
\begin{equation}
\log \mathbf{T} = \sum_{n=1}^{\infty} (-1)^{n+1} 
\frac{(\mathbf{T}-\mathbf{I})^{n}}{n}\, ,
\label{logseries}
\end{equation}
which is the the principal branch of the complex logarithm in
Eq.~(\ref{principallog}) or other numerical methods, such as
Schur decomposition.

To ascertain if the principal logarithm is computable 
one has the following proposition\cite{Israel2001}:
\begin{proposition}
Let $S = \max{\left (
              (a-1)^{2}+b^{2}
              \right )}$ 
where $a$ and $b$ are real coefficients of an eigenvalue $\lambda=a + ib$ 
of the transition matrix $\mathbf{T}$.
If $S < 1$ then the polynomial series of the $\log{\mathbf{T}}$,
Eq.~(\ref{logseries}), converges to a matrix with zero row-sums.
\label{suficient-condition-1}
\end{proposition}

While the existence of the logarithm of a transition matrix is
necessary for our purposes, it does not solve the full embedding
problem. 
One must assure further that a valid generator exists, 
i.e.~a matrix with non-negative off-diagonal entries and zero 
row-sums.
Moreover, it is also true that, if $S>1$, one cannot claim that 
$\mathbf{T}$ has no generator: another generator may exist in a 
different branch.

We are interested in the general case of knowing if there is a valid 
generator, and if there is, to find it.
For that, we need to solve the full embedding problem. 
The full embedding problem comprises a set of propositions which 
are separated in four different categories:
\begin{itemize}
\item[(A)] Conditions for the convergence of the principal logarithm, 
as presented above in Proposition \ref{suficient-condition-1},
that determine if the matrix defined in Eq.~(\ref{logseries})
has finite entries $Q_{kj}$.  
\item[(B)] Necessary conditions for the existence of a generator.
\item[(C)] Sufficient conditions for the existence of a generator.
\item[(D)] Uniqueness conditions of the generator for properly defining
the underlying continuous process.
\end{itemize}

The conditions for the convergence of the principal logarithm are mainly
included in Proposition \ref{suficient-condition-1}.
Most of the other known results, comprehending categories (B), (C) and
(D) are enumerated in the papers by Israel and co-workers\cite{Israel2001} 
and Davies\cite{Davies2010}. In the following we present an
overview of the most relevant propositions.

Regarding the necessary conditions, important for establishing that a 
generator cannot exist, there are three highly used propositions easy to
implement\cite{Israel2001}. The first one is:
\begin{proposition}
If a transition matrix $\mathbf{T}$ obeys one of the following conditions
\begin{itemize}
\item[a)] $\Det{(\mathbf{T})} \leq 0$,
\item[b)] $\Det{(\mathbf{T})} > \prod_{i} T_{ii}$,
\item[c)] $T_{ij}=0$ and there is an integer $n$ such that 
          $(\mathbf{T}^{n})_{ij} \neq 0$,
\end{itemize}
then no valid generator exists.
\label{ncondition1}
\end{proposition}

For $\mathbf{Q} = \log{(\mathbf{T})}$, the equality 
\begin{equation}
\Tr{(\mathbf{Q})} = \log{(\Det{(\mathbf{T})})}
\label{det-exp=exp-tr}
\end{equation}
gives the right insight 
to the property a) in Prop.~\ref{ncondition1} 
since the logarithm of a real number is only defined 
for positive values. 
Property b) is related with the definition of determinant.
As for property c), suppose that a minimum of $t$ transitions are 
needed to go from $i$ to $j$. If the processes is not time-continuous
and transitions do not occur more than once in a time period $\Delta t$, 
then an entity can only go from $i$ to $j$ in a number of transitions
larger than $(t-1)/\Delta t$. 
This naturally is not true for time-continuous processes, 
since there is always a non-zero probability of making $t$ transitions 
between different states over any time-window. 
For a complete proof of Prop.~\ref{ncondition1} see Ref.~\cite{Israel2001}.

The second proposition is:
\begin{proposition}
For a transition matrix $\mathbf{T}$ with distinct eigenvalues, 
a generator $\mathbf{Q}$ exists only if, given any eigenvalue of 
$\mathbf{Q}$ in the form $\lambda=a + ib$, it satisfies the condition 
$|b| \leq \vert \log(\Det{\mathbf{T}})\vert$.
\label{necessary-condition-eigen-3}
\end{proposition}

Proposition \ref{necessary-condition-eigen-3} is related to the previous
one. Consider $\mathbf{T}$ embeddable and define 
$k \equiv \Tr{(\mathbf{Q})} = \log(\Det{(\mathbf{T})})$ (see
Eq.~(\ref{det-exp=exp-tr})). 
All entries of matrix $\mathbf{Q}^{\prime} = \mathbf{Q} - \mathbf{1}k$ are
non-negative and its row-sums are equal to $-k$. 
From Perron-Frobenius Theorem we know that all eigenvalues of
matrix $\mathbf{Q}^{\prime}$ have an absolute value not smaller than $-k$. 
Since $\lambda = a + ib$ is an eigenvalue of $\mathbf{Q}$, then 
$\lambda^{\prime} = (a-k) + ib$ is an eigenvalue of $\mathbf{Q}^{\prime}$,
yielding $k > \vert \lambda^{\prime} \vert > \vert b\vert$.

A third necessary condition defines the region of the complex plane 
that contains the eigenvalues of $\mathbf{T}$, if a generator exists:
\begin{proposition} 
If $\mathbf{T}$ is a $n\times n$ matrix and has a generator, then its 
eigenvalue spectrum is given by $\lambda_k=r_k\exp{(i\theta_k)}$,
where $-\pi \leq \theta \leq \pi$ and 
\begin{equation}
r \leq \exp{\left (
                 -\theta \tan{\left ( \tfrac{\pi}{n} \right )}
                 \right       )} .
\label{eq:nEigenprop1}
\end{equation}
\label{nEigenprop1}
\end{proposition}

The proof of this proposition, and a general description of the 
inverse eigenvalue problem can be found 
in Ref.~\cite{karpelevich1951characteristic,Higham2011}.
It is related with the inverse eigenvalue problem, and can also 
be used when studying the existence of stochastic roots of 
matrices. 

One additional necessary condition for time-homogeneous generators that
will be usefull below when comparing with time-inhomogeneous
generators is the following one:
\begin{proposition}
If $\mathbf{T}$ is embeddable, then every negative eigenvalue of
$\mathbf{T}$ has even algebraic multiplicity. 
\label{nEigenprop0}
\end{proposition}

In general, Prop.~\ref{nEigenprop0} is usefull for the cases when
$\mathbf{T}$ has negative real eigenvalues.

Sufficient conditions for the existence of one
generator, usually deal with considering different branches of the
logarithm of the transition matrix and check if they are valid
generators, i.e., if their off-diagonal entries are real and positive,
and their row-sums are one.
In the particular case when it is known that the only possible
generator is the principal logarithm, then computing
Eq.~(\ref{logseries}) gives a complete answer to  
whether or not a valid generator exists.
In case all necessary conditions hold, it is legitimate to raise the 
hypothesis a generator may exist, but there is still 
the question if the generator is unique.

The following two propositions are sufficient conditions for the uniqueness 
of one homogeneous generator\cite{Israel2001}. The first one reads:
\begin{proposition}
Let $\mathbf{T}$ $\in  \mathbb{R} ^{n \times n}$ be a transition matrix.
\begin{enumerate}
\item[a)] If $\Det{(\mathbf{T})} > \frac{1}{2}$, then $\mathbf{T}$ has at most one generator.
\item[b)]  If $\Det{(\mathbf{T})} > \frac{1}{2}$ and $||\mathbf{T}-\mathbf{I}||<\frac{1}{2}$ using any operator norm, then $\log(\mathbf{T})$ is the only possible generator of $\mathbf{T}$.
\item[c)] If $\mathbf{T}$ has distinct eigenvalues, and $\Det{(\mathbf{T})} > \exp{(-\pi)}$, then $\log(\mathbf{T})$ is the only possible generator of  $\mathbf{T}$. 
\end{enumerate}
\label{unicidade1}
\end{proposition}

The second property b) guarantees that, when there are no repeated eigenvalues,
only a finite number of generators exist. 
Such property is particularly relevant, since in this case it is often 
possible to find all generators\cite{Israel2001}.

The second proposition for the uniqueness of one generator is:
\begin{proposition}
If $\mathbf{T}$ is a Markov matrix with distinct eigenvalues 
$\lambda_{1},\dots,\lambda_{n}$. We have that 
\begin{itemize}
\item[a)] Only a finite number of solution $e^{\mathbf{Q}}=\mathbf{T}$ can be 
Markov Generators. 
\item[b)] If $\vert\lambda_{r}\vert>\exp{(-\pi\tan{(\tfrac{\pi}{n})})}$ 
for all $r$, then the principal logarithm is the only $\mathbf{Q}$ such that 
$\exp{(\mathbf{Q})}=\mathbf{T}$.
\end{itemize}
\label{unicidade2}
\end{proposition}
The proof of both Props.~\ref{unicidade1} and \ref{unicidade2} can be 
found in Ref.~\cite{Israel2001}.

%%%%%
\section{A example: the circulant transition matrix}
\label{sec:example}

As a mathematical problem, the embedding problem is still open 
for a general $n$-dimensional matrix, but it can be analytically
solved for some subclasses of matrices. 
In this section we address in detail a simple example in three dimensions,
namely the embedding of circulant transition matrices of the form:
\begin{equation}
\mathbf{T}_C = 
\left (
\begin{array}{ccc}
  a &  b & c\\
  c &  a & b \\
  b &  c & a \\
\end{array}
\right ) ,
\label{TC}
\end{equation}
or simply $\mathbf{T}_C=\hbox{circ}(a,b,c)$, with $0\geq a,b,c\geq 1$ and $a+b+c=1$.
Circulant transition matrices have two independent degrees of freedom:
any pair of values $(a,b)$ can represent a three-dimensional circulant
transition matrices if $a+b<1$, $a,b>0$. 
See the triangular region in Fig.~\ref{fig01}.

It is easy to check that all necessary conditions in 
Prop.~\ref{ncondition1} for 
a generator to exist are fulfilled if $0<a^3+b^3+c^3-a^3b^3c^3\leq a^3$.
Further, according to Prop.~\ref{necessary-condition-eigen-3} a 
generator may exist if the argument of the eigenvalues of
$\mathbf{T}_C$ are not larger than $\log{(a^3+b^3+c^3-a^3b^3c^3)}$.

For the particular case of the circulant transition matrix,
only Prop.~\ref{nEigenprop1} matters, since in this case it
turns out to be a necessary and sufficient condition as we next prove.
%%%%%%%%%%%%%%%%%%%%%%%%%%%%%%%%%%%%%%%%%%%%%%%%%%%%%%%%%%%%%%%%%%%%%%
\begin{figure}[t]
\centering
\includegraphics[width=0.46\textwidth]{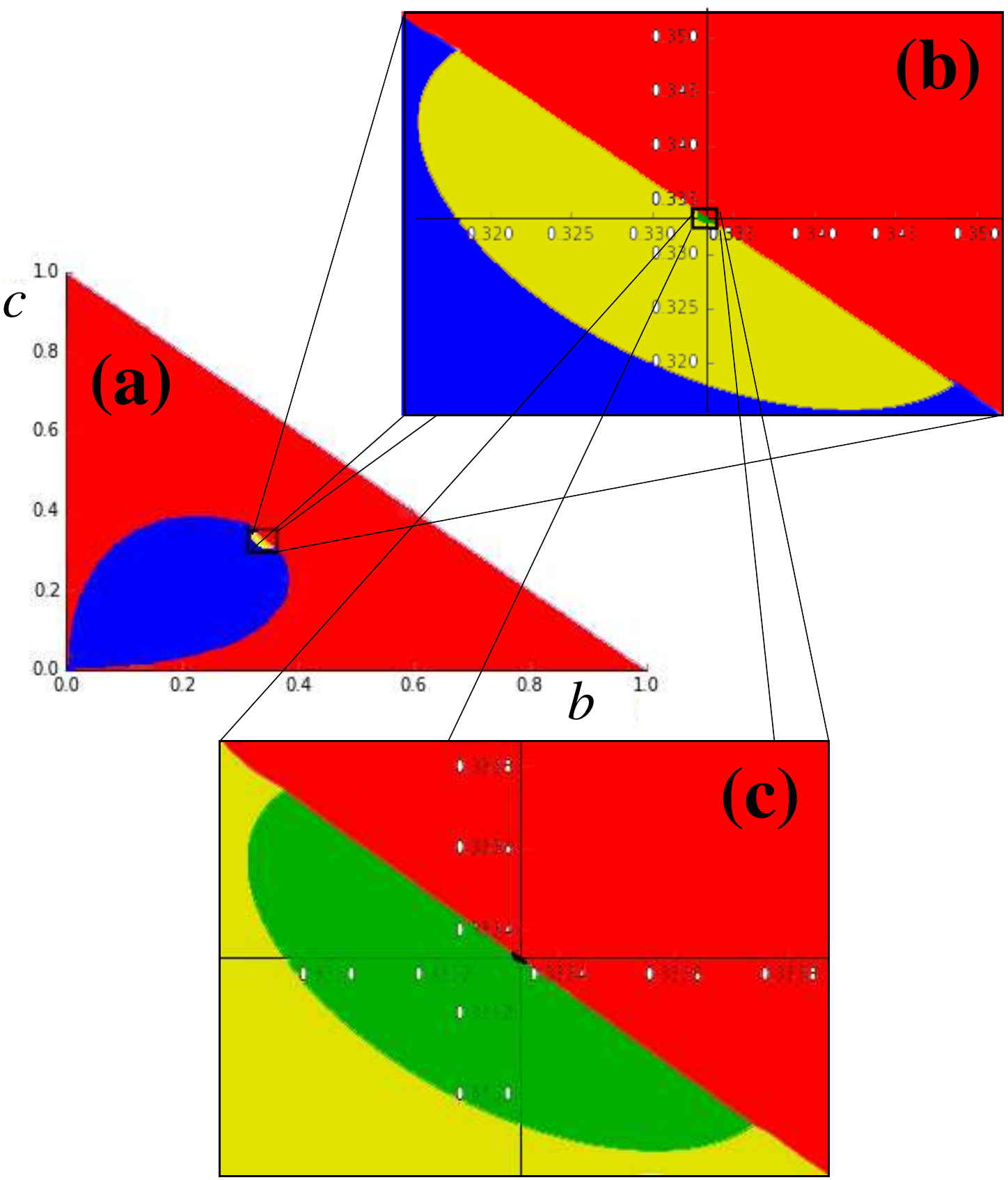}
\caption{%
        (Color online)
        {\bf (a)} 
        Region in parameter-space of transition matrix $\mathbf{T}_C$,
        Eq.~(\ref{TC}), for which a generator $\mathbf{Q}_C$ exists.
        The triangular (red) region is the one for which matrix
        $\mathbf{T}_C$ is a transition matrix, $a,b,c>0$ and $a+b+c=1$.
        The blue region indicates the region of parameter values for
        which only one generator exists, while the yellow (small) 
        region indicates a region where two generators exist.
        {\bf (b)} By zooming in this region shows a green region
        where three generators exist, and
        {\bf (c)} continuing to zoom shows smaller and smaller 
        regions, where a larger number of genertors exist 
        (see text).}
\label{fig01}
\end{figure}
%%%%%%%%%%%%%%%%%%%%%%%%%%%%%%%%%%%%%%%%%%%%%%%%%%%%%%%%%%%%%%%%%%%%%%

To that end, we write the transition matrix as 
$\mathbf{T}_C=\exp{\mathbf{Q}_C}$, since the exponential of a circulant 
matrix with real entries is itself a circulant matrix with real entries 
and consider $\mathbf{Q}_C$
in the form $\mathbf{Q}_C=\hbox{circ}(-\alpha,\beta,\gamma)$.
For $\mathbf{Q}_C$ to be a generator we need to prove that 
$\alpha,\beta,\gamma>0$.

The row-sums of $\mathbf{T}_C$ are equal to one by definition and this
can only happen if the row-sums of $\mathbf{Q}_C$ are equal to zero.
Thus, the equality $\alpha=\beta+\gamma$. 
Moreover, 
it can be shown that, computing the principal logarithm of $\mathbf{T}_C$,
yields a matrix with negative diagonal elements. Thus we take
$\alpha>0$.

Since $\alpha>0$ and all entries of the generator $\mathbf{Q}$ are real 
we need only to prove that $\beta$ and $\gamma$ are both non-negative. 
Since $\alpha=\beta+\gamma$, either $\beta$ or $\gamma$ must be positive. 
Therefore we only need to prove that $\beta \gamma>0$.

Proposition \ref{nEigenprop1} gives a condition for the eigenvalues
of the transition matrix $\mathbf{T}_C$ 
to have a generator matrix.
It can be proven\cite{runnenburg1962elfving} 
that such condition hold if and only if
an equivalent condition for $\mathbf{Q}_C$ holds, namely:
\begin{equation}
\left \vert 
\frac{\Im{(\lambda_i)}}{\Re{(\lambda_i)}} 
\right \vert < \tan(\frac{\pi}{3}) \, ,
\label{ped_start}
\end{equation}
where $\lambda_i$ with $i=1,2,3$ are the eigenvalues of 
$\mathbf{Q}_C$\cite{gray2005toeplitz}:
\begin{subequations}
\begin{eqnarray}
\lambda_1 &=& 0                                      , \label{eig1}\\
\lambda_2 &=& -\beta-\gamma + \beta k + \gamma k^{\ast}                , \label{eig2}\\
\lambda_3 &=& -\beta-\gamma + \beta k^{\ast} + \gamma k = \lambda_2^{\ast}  , \label{eig3}
\end{eqnarray}
\label{ped_eigenvalues}
\end{subequations} 
with $k = e^{\frac{2\pi i}{3}} $ and $k^{\ast}$ its complex conjugate. 

Using $\lambda_2$ in Eq.~(\ref{eig2}) and substituting in 
Eq.~(\ref{ped_start}), yields
\begin{equation}
\left \vert 
\frac{\lambda_2 - \lambda_2^{\ast}}{\lambda_2 + \lambda_2^{\ast}} 
\right \vert 
< \tan{\left ( \frac{\pi}{3} \right ) } \, 
\label{ped_start2}
\end{equation}
and through algebraic manipulation one arrives to
\begin{equation}
\left \vert \frac{\beta-\gamma}{\beta+\gamma} \right \vert < 1 .
\label{proof1}
\end{equation}
The last inequality implies necessarily that $\beta \gamma > 0$.
A similar result is obtained by substituting in 
Eq.~(\ref{ped_start}) one of the other eigenvalues
$\lambda_0$ and $\lambda_2$.

Hence, in our particular case of a circulant matrix, Prop.~\ref{nEigenprop1} 
is also a sufficient condition and one needs only to determine the
inequality in Eq.~(\ref{eq:nEigenprop1}) as a function of the degrees
of freedom in matrix $\mathbf{T}_C$ for all its three eigenvalues 
\begin{subequations}
\begin{eqnarray}
\lambda^{(T)}_1 &=& 1                                      , \label{eig1T}\\
\lambda^{(T)}_2 &=& \tfrac{1}{2}(2-3b-3c)+\tfrac{\sqrt{3}}{2}(b-c)i     , \label{eig2T}\\
\lambda^{(T)}_3 &=& \tfrac{1}{2}(2-3b-3c)-\tfrac{\sqrt{3}}{2}(b-c)i     , \label{eig3T}
%(\lambda^{(T)}_2)^{\ast}  . \label{eig3T}
\end{eqnarray}
\label{ped_eigenvaluesT}
\end{subequations}

The first eigenvalue is independent of the parameters.
The
other two are complex conjugate, having the same norm $r$ and
symmetric arguments $\theta$.
Thus, we only need to consider one
eigenvalue, say $\lambda^{(T)}_3=r\exp{(i\theta)}$, which according to
Prop.~\ref{nEigenprop1} for $\mathbf{T}_C$ to be embeddable
must fullfil $r\leq \exp{(-\sqrt{3}\theta)}$ with
\begin{equation}
r=\tfrac{1}{2}\left (
    \left ( 2-3(b+c) \right )^2 + 3(b-c)^2
    \right )^{1/2}
\label{rcirc}
\end{equation}
and
\begin{equation}
\theta = \left \{
\begin{array}{lcl}
\arctan{\tilde{\theta}}  &\Leftarrow& c<\tfrac{2}{3}-b , \cr
& & \cr
\arctan{\tilde{\theta}} + \sgn{(b-c)}\pi  &\Leftarrow&
c>\tfrac{2}{3}-b , \cr
& & \cr
\tfrac{\pi}{2}\sgn{(b-c)}  &\Leftarrow& c=\tfrac{2}{3}-b , 
\end{array}
\right .
\label{rcircArg}
\end{equation}
where $\tilde{\theta}=\sqrt{3}(b-c)/(2-3b-3c)$.

Figure \ref{fig01} shows the region within the triangle $1-b-c>0$,
$b>0$ and $c>0$ where the circulant transition matrix $\mathbf{T}_C$ 
has a generator, i.e.~the region where $a=1-b-c$ and $b$ and $c$
obey the inequality in Eqs.~(\ref{eq:nEigenprop1}), (\ref{rcirc}) and (\ref{rcircArg}).
The number of valid generators of $\mathbf{T}_C$, a three-dimensional
circulant transition matrix, can also be determined from its
eigenvalues, namely it is given by the largest integer smaller than 
$\left ( \sqrt{3} 
            \log{\left (
                    \Re^2{(\lambda^{(T)})} +
                    \Im^2{(\lambda^{(T)})} 
                    \right )} \right )/ (4\pi)$.

Figure~\ref{fig01}a shows one blue region and one smaller yellow
region.
While the blue region indicates the set of parameter values for $b$
and $c$ for which only one generator exists, the yellow region 
comprehends the set of parameter values for which $\mathbf{T}_C$ has
two or more generators.
By zooming in this region, smaller and smaller regions appear,
Figs.~\ref{fig01}b and \ref{fig01}c, near the crossing point between 
the diagonal $c=b$ and the line $c=\tfrac{2}{3}-b$, 
where a larger number of generators exist.

%%%%%%%%%%%%%%%%%%%%%%%%%%%%%%%%%%%%%%%%%%%%%%%%%%%%%%%%%%%%%%%%%%%%%%%
\section{The time-inhomogeneous embedding problem}
\label{sec:inhomogeneous}

In this section we show which of the known theorems for time-homogeneous 
embedding problem hold when both transition matrix and its generator 
depend explicitly on time. In this scope, we provide three new conditions,
two necessary and one sufficient, for the existence of a
time-inhomogeneous generator.
We also provide two additional necessary and sufficient conditions which 
enable the possibility for testing equivalent matrices.

The generator $\mathbf{Q}(t)$ is considered to explicitly depend
on time $t$, as well as its 
corresponding transition matrix $\mathbf{T}(t,\tau)$. 
As stated in the introduction, a transition matrix is solution of 
Eq.~(\ref{MC1}),
i.e.~it has a generator if and only if it describes a time-continuous 
and Markov process, besides having the properties of a transition
matrix (non-negative entries and unitary row-sums).

For time-inhomogeneity, the general solution of Eq.~(\ref{MC1}) is 
given by:
\begin{equation}
\mathbf{T}(t,\tau)= \sum_{k=0}^{\infty} X_{k}(t-\tau)\, ,
\label{general_solution}
\end{equation}
with $X_{0}(t-\tau) \equiv X_0 = \mathbf{1}$ and 
\begin{equation}
X_{k+1}(t-\tau) = \int_{t-\tau}^{t} X_{k}(s)\mathbf{Q}(s) ds \, . 
\label{general_solution2}
\end{equation}
Equation (\ref{general_solution2}) is known as the Peano-Baker 
series\cite{baake2011peano}.
In the particular case that $\mathbf{Q}(t)$ and $\mathbf{Q}(t^{\prime})$ 
commute for all $t$ and $t^{\prime}$ solution (\ref{general_solution}) 
reads
\begin{equation}
\mathbf{T}(t,\tau) = \exp{\left (
                      \int_{t-\tau}^{t}\mathbf{Q}(s)ds
                      \right )} \, .
\label{general_solution3}
\end{equation}

The first necessary proposition for time-inhomogeneous generators
follows simply from the fact that 
$\mathbf{T}(t,\tau)$ is a transition matrix:
\begin{proposition}
If a transition matrix $\mathbf{T}(t,\tau)$ has a negative 
determinant, then no generator $\mathbf{Q}(s)$ exists,
for $t<s<t+\tau$.
\label{necessary1}
\end{proposition}
\begin{proof}
To prove the positiveness of the determinant we start by assuming that 
a generator $\mathbf{Q}(t)$ exists. Then, 
letting the arguments of $\mathbf{T}$ and $\mathbf{Q}$ drop for simplicity,
\begin{subequations}
\begin{eqnarray}
\frac{d}{dt} \Det{\mathbf{T}} &=& 
\Det{\mathbf{T}}\Tr{\left (
                                \mathbf{T}^{-1}
                                \frac{d\mathbf{T}}{dt}
                                \right )}  \label{proof2_0}\\
\frac{d \log \left ( \Det{\mathbf{T}} \right)}{dt} &=& 
\Tr{\left (
   \mathbf{T}^{-1}\mathbf{T}\mathbf{Q}
   \right )} \label{proof2_00}\\
\Det{\mathbf{T}} &=& 
\exp{\left ( \int_{t-\tau}^{t} \Tr{\mathbf{Q}} ds \right )} > 0 .
\label{proof2}
\end{eqnarray}
\end{subequations}
The final inequality stands true since the trace of $\mathbf{Q}(t)$ is always 
a real (negative) value.
\end{proof}

The second necessary proposition deals also with the fact that $\mathbf{T}$ 
is a transition matrix, namely that its entries are probabilities:
\begin{proposition}
If a transition matrix $\mathbf{T}$ fulfills 
$\Det{\mathbf{T}} > \prod_{i} T_{ii}$, then no generator exists.
\label{necessary2}
\end{proposition}

\begin{proof}
If $\mathbf{T}$ has a generator, then,
%\begin{subequations}
\begin{eqnarray}
\frac{dT_{kk}(t, \tau)}{dt}  &=&    \sum_{j} Q_{kj}(t + \tau) T_{jk}(t, \tau)\, , 
\label{eq_difnec2}
\end{eqnarray}
and, since for $k\neq j$, $T_{kj} > 0$ and $Q_{kj} < 0 $, one
arrives to
\begin{eqnarray}
\frac{dT_{kk}(t,\tau)}{dt}  &\geq&   Q_{kk}(t + \tau)T_{kk}(t, \tau) .
\label{eq_difnec22}
\end{eqnarray}
Since $T_{kk}(t,0) = 1$, we can integrate the differential equation in 
Eq.~(\ref{eq_difnec2}) yielding
\begin{equation}
T_{kk}(t, \tau) \geq \exp{\left (
                     \int_{t-\tau}^{t} Q_{kk}(s) ds
                         \right )} ,
\label{eq_nec2}
\end{equation}  
where Gr\"onwall's inequality is used\cite{pachpatte98},
and finally, from Eq.~(\ref{proof2}), one arrives to
\begin{eqnarray}  
\prod_{k} T_{kk}(t, \tau) &\geq& \prod_{k} \exp{ \left (
                          \int_{t-\tau}^{t} Q_{kk}(s) ds \right )} \cr
                &=&  \exp{\left (
                          \sum_k \int_{t-\tau}^{t} Q_{kk}(s) ds
                          \right ) } \cr
                &=&  \exp{\left (
                          \int_{t-\tau}^{t} \Tr{(Q_{kk}(s))} ds
                          \right )} \cr
                &=& \Det{(\mathbf{T}(t,\tau))}  .
\end{eqnarray} 
\end{proof}

The sufficient condition we will implement afterwards deals with the 
particular case of a LU decomposition:

\begin{proposition}
If $\mathbf{T}$ has a LU decomposition with  $\mathbf{L}$ and $\mathbf{U}$ having only non-negative elements, then $\mathbf{T}$ has an inhomogeneous generator $\mathbf{Q}(t)$.  
\label{LU}
\end{proposition}

\begin{proof}
To prove this proposition it is important to know an auxiliary result, Prop.~\ref{prop:append1} in Append.~\ref{append:results}, from which it follows that the property of having a time-dependent generator is preserved under multiplication.
We use this results from proving that a matrix having an $LU$ decomposition, with $\mathbf{L}$ and $\mathbf{U}$ with non-negative entries, can be modeled through a time-dependent generator. For that, it suffices to prove that the matrix $\mathbf{T}$ has an LU decomposition with $\mathbf{L}$ and $\mathbf{U}$ transition matrices. 

Let us first define a diagonal matrix $\mathbf{D}$ with entries
$D_{ii} = (\sum_j U_{ij})^{-1}$.
Thus, $\mathbf{T}$, with dimension $n\times n$ can be written as
$\mathbf{T} =  \mathbf{L} \mathbf{U} =\mathbf{L} \mathbf{D}^{-1} 
\mathbf{D} \mathbf{U} = \mathbf{L}^{\prime} \mathbf{U}^{\prime}$,
with $\mathbf{L}^{\prime} =  \mathbf{L} \mathbf{D}^{-1}$  and  $\mathbf{U}^{\prime} =  \mathbf{D} \mathbf{U}$ triangular matrices that have all non-negative elements since they are a multiplication of one diagonal matrix with one triangular matrix, all of them with non-negative elements. Furthermore their row-sums are one, since
\begin{eqnarray}
\sum_j U^{\prime}_{ij} &=&  \sum_j \sum_k D_{ik} U_{kj} \cr
          & =&  \sum_k D_{ik} ( \sum_j U_{kj} ) = D_{ii} \sum_j U_{ij} ) \cr
          &= & ( \sum_k  U_{ik})^{-1} ( \sum_j U_{kj} ) = 1  ,
\end{eqnarray}
for all $i=1,\dots,n$.
Analogously, since $\sum_j T_{ij} =  1$ for $i$, one has 
\begin{eqnarray}
 \sum_j T_{ij} &=& \sum_j \sum_k L^{\prime}_{ik} U^{\prime}_{kj} \cr  
   &=&  \sum_k L^{\prime}_{ik} ( \sum_j U^{\prime}_{kj} ) \cr
   &=&  \sum_k L^{\prime}_{ik} =  1 .
\end{eqnarray} 
and therefore
\begin{eqnarray}
\sum_k L^{\prime}_{ik} &=&  \underline{1}  .
\end{eqnarray} 
\end{proof}

Notice that in the $LU$ factorization there are usually $n^2 + n$ variables and $n^2$ equations. By imposing the row-sums equal to one, we get $n^2 + n$ equations, and consequently the $LU$ decomposition defined in this way is unique.

To end this Section we introduce two additional propositions, which are 
necessary 
and sufficient for both time-homogeneous and inhomogeneous cases.
They are useful when implementing the computational framework for detecting
generators, since they help to handle cases
where the application of the above propositions do not provide satisfactory
output for the embedding problem.
With these equivalent matrices one aims to derive a class of matrices 
that are embeddable if and only if the ``original'' transition matrix 
$ \mathbf{T} $ is embeddable, which expands the set of possible matrices 
one may properly test.

The first proposition uses the similarity of matrices through permutation 
matrices: 
\begin{proposition}
Let $\mathbf{A} = \mathbf{P}^\top \mathbf{T} \mathbf{P}$, where
$\mathbf{P}$ is a permutation matrix and $\mathbf{T}$ is a transition
matrix. $\mathbf{T}$ is embeddable if and only if $\mathbf{A}$ is also
embeddable. 
\label{basis-permutation}
\end{proposition}
\begin{proof}
To prove this proposition, we will consider a reblabeling of
the states $i$, $j$, etc. Notice that, under such relabeling,
the properties of the transition matrix do not change. 
Therefore, since changing the transition matrix $\mathbf{T}$ by 
$\mathbf{P}^\top \mathbf{T} \mathbf{P}$ one is, in fact, just relabeling 
the states, one intuitively concludes that if $\mathbf{T}$ is embeddable, 
then $\mathbf{P}^\top \mathbf{T} \mathbf{P}$ should also be embeddable.

We start by assuming that $\mathbf{T}$ is embeddable,
\begin{equation}
\mathbf{T}=\exp{(\mathbf{Q})} = \sum \frac{\mathbf{Q}^n}{n!} ,
\end{equation} 
where $\mathbf{Q}$ is the generator of $\mathbf{T}$.
Since $e^{\mathbf{P}^\top \mathbf{Q} \mathbf{P}} = \mathbf{P}^\top
e^{\mathbf{Q}} \mathbf{P} = \mathbf{P}^\top \mathbf{T} \mathbf{P}$, we
only need to prove that $\mathbf{Q}^{\prime} = \mathbf{P}^\top
\mathbf{Q} \mathbf{P}$ is a valid generator, i.e.~it must have zero
row-sums and non-negative off-diagonal entries. 

Since $\mathbf{Q}$ is a valid generator one has
\begin{eqnarray}
\sum_j Q^{\prime}_{ij} &=& \sum_j \sum_k \sum_l P^\top_{il} Q_{lk} P_{kj} \cr
   &=& \sum_k \sum_l P^\top_{il} Q_{lk} (\sum_j P{kj}) \cr
   &=& \sum_k \sum_l P^\top_{il} Q_{lk} \cr
   &=&  \sum_l P^\top_{il} \sum_k Q_{lk}   \cr
   &=&  \sum_l P^\top_{il}\times 0 = 0 \, .
\end{eqnarray}

To prove that matrix $\mathbf{Q}^{\prime}$ has non-negative off-diagonal
entries we write for $k \neq l$ the off-diagonal entry
$Q_{kl}^{\prime} = \sum_{nm} P_{kn} Q_{nm} (P^{\top})_{ml}$ and note that,
since the matrix $\mathbf{P}$ has only one non-zero element per column 
and per row. Thus, being that column $i$ and row $j$, one has
$Q_{kl}^{\prime} = P_{ki}Q_{ij} (P^{\top})_{jl}$. 

If $k \neq l$ and $i=j$, then $ P_{ki} = 1$ and $ (P^{\top})_{il} =
P_{li} = 1$ which contradicts the fact that $\mathbf{P}$ is a
permutation matrix. Thus, if $k \neq l$ then $i\neq j$, and so 
there is a direct correspondence between off-diagonal
elements of $\mathbf{Q}^{\prime}$ and those of $\mathbf{Q}$:
if all $Q_{ij}$ are non-negative so are all $Q_{kl}^{\prime}$.

Conversely, if $\mathbf{A}$ is embedabble, one just writes
$\mathbf{T}=(\mathbf{P}^\top)^{-1}\mathbf{A}\mathbf{P}^{-1}=
(\mathbf{P}^\prime)^\top\mathbf{A}(\mathbf{P}^{\prime})^\top$
with $\mathbf{P}^\prime=\mathbf{P}^{-1}$ and applies the same arguments
as above.
\end{proof}

The second proposition uses renormalization and transposition of the
``original'' transition matrix:
\begin{proposition}
Let $\mathbf{T}$ be a transition matrix and consider $\mathbf{B} =
\mathbf{D} \mathbf{T}^\top $, where $\mathbf{D}$ is the diagonal
matrix $D_{ii} = ( \sum_j T^\top_{ij} )^{-1}$. 
$\mathbf{T}$ is  embeddable if and only if $\mathbf{B}$ is also embeddable.
\label{time-reversion}
\end{proposition}
\begin{proof}
It is easy to see that if $\mathbf{T}$ is a transition matrix so is
$\mathbf{B}$, since $\mathbf{B}$ is always normalized to have row-sums
one, and if $\mathbf{T}$ has all its elements non-negative, so has
$\mathbf{B}$. 
Notice that, while $T$ yields the probabilities to which 
a present states transitates, $B$ gives the probabilities {\it from} which
a state has transitated.
It was proven that for a fixed time $t$, a matrix $\mathbf{T}(t,\tau)$
has all its entries non-negative, $T_{ij}(t,\tau) > 0$, for all $\tau$
and is time-continuous, i.e.~for any $\epsilon>0$ there is one
$\delta$ for which, if $\vert  \tau_1 - \tau_2 \vert < \delta$ then
$\vert \vert \mathbf{T}(t,\tau_1) - \mathbf{T}(t,\tau_2) \vert \vert <
\epsilon$ if and only if there is a valid generator associated with
$\mathbf{T}(t,\tau)$. 
Since $\mathbf{B}$ is the product of two matrices that are time
continuous, $\mathbf{B}$ is also time-continuous.
\end{proof}

%%%%%%%%%%%%%%%%%%%%%%%%%%%%%%%%%%%
\section{Computational implementation: how ``embeddable'' is a matrix?}
\label{sec:implementation}

The mathematical conditions for the existence of a homogeneous 
generator from the embedding problem are useful more at a theoretical 
than at a computational level. 
They give a bivalent result that does not take into consideration 
neither noise generated from finite samples nor how distant an empirical 
process is from having a constant generator.

In this section we will describe how to adapt our mathematical results
to be meaningful to empirical transition matrices in real situations.
First, we evaluate how embeddable a transition matrix is, we define in 
Sec.~\ref{subsec:metrics} a proper metric for each proposition above
that measures how ``close'' the empirical transition matrix is from 
satisfying the corresponding proposition.
Then, in Sec.~\ref{subsec:Qestimate}, if 
one arrives to the conclusion that the transition 
matrix is indeed embeddable we describe proper ways to model its corresponding 
generator.

There are several differences between the time-homogeneous and the
time-inhomogeneous problem:
\begin{itemize} 
\item In the time-inhomogeneous problem there is no finite set of possible
generators, as is usually the case in the time-homogeneous
counterpart, namely when the transition matrix has no repeated 
eigenvalues\cite{Israel2001}. 
If there is a non-homogeneous generator, then
there is an infinite number of them.
\item The product of two homogeneous embeddable matrices might not be
time-homogeneous embeddable, whereas the product of two
time-inhomogeneous matrices is always embeddable.
\item In the inhomogeneous case, the existence of a real-valued
  logarithm is not a necessary condition for being embeddable.
\item The necessary conditions of the time-homogeneous problem
  concerning the eigenvalues of the transition matrix,
  Props.~\ref{nEigenprop1} and \ref{nEigenprop0},  are not necessary
  conditions for the time-inhomogeneous problem.  
\end{itemize}

As an illustrative example consider the matrix:
\begin{equation}
\mathbf{T}_E=\left [
\begin{array}{ccc}
0.1179 & 0.0890 & 0.7931 \\
0.0100 & 0.1000 & 0.8900 \\
0.8901 & 0.0010 & 0.1089 \\
\end{array}
\right ],
\end{equation}
The matrix $\mathbf{T}_E $ is, according to Prop.\ref{LU},
time-inhomogeneous embeddable, since it
is a product of matrices that have a positive LU decomposition.
However it is not time-homogeneous embeddable, since 
it has distinct negative eigenvalues, $\{1, -0,001490, -0,671710 \}$,
and thus it has no real-valued logarithm\cite{higham2008functions}.
Moreover, the conditions in both Props.~\ref{nEigenprop1} and
\ref{nEigenprop0} are not fulfilled.

Regarding Prop.~\ref{ncondition1}, we have shown that conditions 
a) and b) are necessary conditions for the more general case of 
time-inhomogeneous generators.
As for condition c), one can show that there is also 
a limit number of zero entries for the time-inhomogeneous case.
See Prop.~\ref{necessary3} in Append.~\ref{append:results}.

%Finally, we could not assert whether or not 
%Prop.~\ref{necessary-condition-eigen-3} is also a necessary 
%condition for the time-inhomogeneous problem.

Before proceeding, two important remarks.
First, it is necessary to describe
how to estimate the transition matrix 
directly from data series and then explain how to resample the transition 
matrix which will be necessary for evaluating if it is embeddable or not. 
Among several algorithms\cite{Burrige1995,jafry2004measurement},
we concentrate in the so-called "Cohort Method", which
counts the number $N_{kj}$  of transitions from state $k$ 
to state $j$ in the desired time-interval $[t,t+\tau]$, 
defining the entries of the transition matrices as
\begin{equation}
T_{kj}(t) = \frac{N_{kj}(t)}{\sum_j N_{kj}(t)}\, ,
\label{cohort_method}
\end{equation} 
with the associated error
\begin{equation}
\sigma_{T_{kj}} = \sqrt{\frac{T_{kj}(1-T_{kj})}{N_{kj}}}.
\label{cohort_method_err}
\end{equation}

Second, in order to implement the set of propositions with an
associated statistical error,
we propose a method of resampling a given empirical transition matrix 
$\mathbf{T}(t, \tau)$.
The set of resampling matrices obtained is then used to quantify the 
error associated to the estimates on the transition matrix:
each metric that is applied to the empirical transition matrix
retrieves a set of metric values when applied to the full set of resampling
matrices, and the standard deviation of that value distribution is then 
taken as the error or uncertainty associated to the metric estimation.

More specifically, one generates number series from the distribution of
states $P(X,t)$ at time $t$ till the distribution $P(X,t+\tau)$ at $t+\tau$,
and estimates the corresponding resampling matrix
through the Cohort Method. See Eq.~\eqref{cohort_method}.

%%%%%%%%%%
\subsection{Embeddability metrics}
\label{subsec:metrics}

The propositions of the embedding problem do not take in consideration the 
uncertainty in the estimation of $\mathbf{T}$, and thus we need to develop 
methods to determine, beyond statistical uncertainty, whether a generator 
exists or not. 
Notice that the embedding problem determines only if the process can be 
modeled as a time-continuous Markov process, but it cannot guarantee if 
the underlying process actually is a time-continuous Markov process. Thus 
we will use a proper null hypothesis, which if not rejected, one assumes
that a suitable generator can be estimated.
In the case of Props.~\ref{necessary1} and \ref{necessary2} the null 
hypothesis states that a generator exists, while for Prop.~\ref{LU} the 
null hypothesis states that such a generator does {\it not} exist.
The null hypothesis is tested for each proposition separately.

To evaluate if the condition of Prop.~\ref{necessary1} is fulfilled 
for a given transition matrix $\mathbf{T}$, we compute the following 
quantity,
\begin{equation}
d_{N_1} = -\frac{\det(\mathbf{T})}{\sigma_{det}} \, ,
\label{distanceN1}
\end{equation}
where $\sigma_{det}$ is the standard deviation from the sample of the 
determinants calculated for each resampling matrix.
If $d_{N_1} \geq 2$, we assume that the determinant of $\mathbf{T}$ is 
negative and the distribution of the resampled determinants are all negative 
within two standard deviations. In this case we reject the null hypothesis,
i.e.~no generator exists.

Regarding the condition in Prop.~\ref{necessary2}, we use the following 
metric,
\begin{equation}
d_{N_2}= -\frac{\prod_{i} T_{ii} - \det(\mathbf{T})}
              {\sigma_{prod}+\sigma_{det}} \, ,
\label{distanceN2}
\end{equation}
where $\sigma_{prod}$ is the standard deviation 
associated to the variable $\prod_i T_{ii}$
according to the expression in Eq.~(\ref{cohort_method_err}).
Again, if $d_{N_2} \geq 2$, then no generator exists. 
%\footnote{$\hat p \pm z \sqrt{\frac{1}{n}\hat p \left(1 - \hat p \right)}$}

Concerning the sufficient condition
of the $LU$-decomposition with non-negative elements, 
Prop.~\ref{LU}, we can use the following distance:
\begin{equation}
d_{S_1}= \min \{ m_L, m_U \} \, ,
\label{distanceS1}
\end{equation} 
with
\begin{subequations}
\begin{eqnarray}
m_L &=& \min_{i,j} \left \{ \tfrac{L_{ij}}{\sigma_{L_{ij}}} 
                  \right \} ,\label{minL}\\
m_U &=& \min_{i,j} \left \{ \tfrac{U_{ij}}{\sigma_{U_{ij}}} 
                  \right \} ,\label{minU}
\end{eqnarray}
\end{subequations}
where $L_{ij}$ and $U_{ij}$ represent the entries of the matrices $\mathbf{L}$ 
and $\mathbf{U}$  respectively, and $\sigma_{Lij}$ and $\sigma_{Uij}$ their
corresponding standard deviations. 
The quantities $\sigma_{L_{ij}}$ and $\sigma_{U_{ij}}$ are calculated solving 
the same system of equations of the $LU$ decomposition, but using the 
uncertainties in the estimation of $T_{ij}$ with the rules of error 
propagation.
If $d_{S_1}>2$, then we reject the null hypothesis, i.e.~we assume that 
a generator exists.

Applying these three metrics to one transition matrix, if the null hypothesis
cannot be rejected, we estimate a generator matrix as describe in the 
Sec.~\ref{subsec:Qestimate}. 
To ascertain if the estimated generator matrix yields a transition matrix
sufficiently close to the empirical transition matrix, we use it
to generate auxiliary transition matrices $\tilde{\mathbf{T}}$.
If the auxiliary matrices are typically close to the empirical 
transition matrix $\mathbf{T}$ we assume that the estimate is
good.
To that end, we introduce one additional metric to assert if the matrix 
$\mathbf{T}$ is close enough to a auxiliary matrix, $\tilde{\mathbf{T}}$,
originated from a time-continuous Markov process with a generator 
$\mathbf{Q}(t)$, is to compute the quantity,
\begin{equation}
d_{est} = \frac{1}{R} \sum_{k=1}^{R} 
         \Theta\left ( 
                 \vert\vert 
                 \mathbf{T}-\tilde{\mathbf{T}}
                 \vert\vert_F 
                 -
                 \vert\vert  
                 \mathbf{T}^{\prime}-\tilde{\mathbf{T}}
                 \vert\vert_F 
               \right ) \, ,
\label{distanceNS1}
\end{equation} 
where $R$ is the number of auxiliary matrices,
$\Theta(x)$ is the Heaviside function and
$\vert\vert \mathbf{X} \vert\vert_F=\left ( \sum_{i=1}^n \sum_{j=1}^n
  X_{ij}^2 \right )^{1/2}$
is the Frobenius norm of matrix $\mathbf{X}$.
We assume that the empirical process, observed for the estimation 
$\mathbf{T}^{\prime}$ is not close to the time-continuous Markov process 
with a transition matrix $\tilde{\mathbf{T}}$ if $d_{est} < 0.10$, i.e.~if 
less than $10\%$ of the auxiliary matrices are outside a confidence interval
with significance value $p=d_{est}$.

If the distance $d_{est}$ is too small a new matrix is generated within 
the conditions of Props.~\ref{basis-permutation} and \ref{time-reversion}. 
In case that the new matrices pass the tests above, these propositions 
guarantee that the original matrix also passes.
%%%%%%%%%%%%%%%%%%%%%%%%%%%%%%%%%%%%%%%%%%%%%%%%%%%%%%%%%%%%%%%%%%%%%%%
\begin{table}
\begin{center}
\begin{tabular}{|c|ccc|} 
\hline 
Metric & $d_{N1}$ & $d_{N2}$ & $d_{S1}$ \\
       & (Prop.\ref{necessary1}) & (Prop.~\ref{necessary2}) & (Prop.~\ref{LU}) \\ 
\hline\hline 
$>2$   & 166 & 199 & 178 \\ 
$<2$   & 34 & 1 & 22 \\ 
\hline
\end{tabular}
\end{center}
\caption{\protect 
         Test results of the inhomogeneous framework detection for a 
         set of $200$ samples, each one with $10^4$ points.
         When one of the metrics is larger than two, the null hypothesis
         cannot be rejected.}
\label{tab01}
\end{table}
%%%%%%%%%%%%%%%%%%%%%%%%%%%%%%%%%%%%%%%%%%%%%%%%%%%%%%%%%%%%%%%%%%%%%%%

To test all the metrics introduced above we generate a set of $200$
samples of $10^4$ points, each one from a different inhomogeneous 
transition matrix, as described below.
We then compute numerically the transition matrix
from each sample and apply all three metrics $d_{N1}$, $d_{N2}$
and $d_{S1}$.
The results, summarized in Tab.~\ref{tab01}, clearly show that
in at least $80\%$ of the cases the framework is able to 
correctly detect the inhomogeneity of an existing generator.

%%%%%%%%%%%%%%%%%%%%%%%%%%%%%%%%%%%%%%%%%%%%%%%%%%%%%%%%%%%%%%%%%%%%%%%%%
\subsection{Modeling the generator matrix $\mathbf{Q}(t)$}
\label{subsec:Qestimate}

In case the null hypothesis cannot be rejected (i.e.~that a valid 
generator exists), we then derive an estimate 
$\mathbf{Q}(t)$ able to model the empirical process. Unlike the case of 
the time-homogeneous embedding problem, here we need to estimate a matrix 
which changes in time and therefore a different procedure is necessary.
In general, for deriving an inhomogeneous generator, one solves the 
Peano-Baker series Eq.~(\ref{general_solution}). 
Assume $\mathbf{Q}(t)$ can be modeled as a polynomial of degree $N$, i.e. 
\begin{equation}
\mathbf{Q}(t) = \sum_n^N \mathbf{B}_n t^n \, ,
\label{poly_gen}
\end{equation} 
where each matrix $\mathbf{B}_n$ is constant over time. 
Naturally, we need to make sure that no off-diagonal entry in 
$\mathbf{Q}(t)$ ever become negative in $t \in [0,1]$.
Introducing Eq.~(\ref{poly_gen}) in Eq.~(\ref{general_solution}) 
yields
\begin{equation}
\mathbf{T} = \sum_{k=0}^{\infty} \prod_{l=1}^{k} \sum_{n=0}^N  \frac{\mathbf{B}_n}{l+\sum_{m=1}^l n_m} .
\label{poly_gen_forT}
\end{equation} 
To invert Eq.~(\ref{poly_gen_forT}) however is very cumbersome and 
computationally expensive.
In this subsection, we propose an alternative for estimating inhomogeneous 
generators that is accurate and easily implementable. 

%Alternatively, we can expand $\mathbf{Q}(t)$ in a fourier series,
%\begin{equation}
%\mathbf{Q}(t) = \sum_n^N \mathbf{c}_n e^{2\pi i n t} \, ,
%\label{poly_gen2}
%\end{equation} 
%but the results are messier and it is harder to check if the generator 
%has negative off-diagonal entries, so probably it is not worthy it.

Our procedure is based in the assumption that the original transition
matrix is a product of a finite number of embeddable matrices, 
$\mathbf{T} = \mathbf{T}_1^{\ast} \dots \mathbf{T}_n^{\ast}$ 
with each $\mathbf{T}_i^{\ast}$ ($i=1,\dots,n$) having an homogeneous generator.

One starts with a decomposition of the form
\begin{equation}
\mathbf{T} = \mathbf{A}_1   \dots \mathbf{A}_n \mathbf{T}_0^{\ast} 
             \mathbf{A}_{n+1} \dots \mathbf{A}_{2n} , 
\end{equation}
where $\mathbf{A}_i$ are embeddable
matrices having one off-diagonal positive term.
The objective here is to find an embeddable matrix $\mathbf{T}_0^{\ast}$ 
from the empirical matrix $\mathbf{T}$ through the multiplication by 
matrices $\mathbf{A}_i$.
If $\mathbf{T}=e^{\mathbf{Q}}$ and $\mathbf{Q}$ has one negative off-diagonal 
entry, $Q_{ij}< 0$, we can try ``correct'' that entry by multiplying 
$\mathbf{T}$ by two matrices, $\mathbf{A}_l$ and $\mathbf{A}_{l+n}$, such 
that $(A_l)_{ik} > 0$ and  $(A_{l+n})_{kj} > 0$ for a fixed index $k$.
Intuitively, if there are transitions from a state $k$ to a state 
$j$ and only afterwards from another state $i$ to state $k$, 
a time-inhomogeneous process might correspond
to a logarithm with a negative off-diagonal entry if $Q_{ij} < 0$.
Hence, one derives a first estimate $\mathbf{T}_0^{\ast}$ of the transition
matrix $\mathbf{T}$.
In case there is more than one negative off-diagonal element of $\mathbf{Q}$
one proceeds similarly for each element separately.

The algorithm proceeds then as follows:
\begin{enumerate}
\item Compute $\mathbf{Q}_0^{\ast} = \log{\mathbf{T}_0^{\ast}}$ and verify it 
is a valid generator. Note that, during the algorithm we must always use 
the same branch of the complex logarithm. 
\item If the generator is not valid, i.e.~it has at least one negative
off-diagonal entry 
$(Q_0^{\ast})_{ij}$, one finds a suitable integer $k$ for which two 
matrices, $\mathbf{A}_1$ and $\mathbf{A}_{n+1}$, have entries 
$(A_1)_{kj}>0$  and $(A_{n+1})_{ik}>0$. 
\item One considers the new estimate $\mathbf{T}_1^{\ast}=\mathbf{A}_1
\mathbf{T}^{\ast}_0\mathbf{A}_{n+1}$ and computes the generator estimate 
$\mathbf{Q}_1^{\ast} = \log{\mathbf{T}_1^{\ast}}$ and verifies if it is 
now a valid generator. 
\item One proceeds recursively until for a certain recursive step $i$
$\mathbf{Q}_i^{\ast} = \log{\mathbf{T}_i^{\ast}}$ has no negative
off-diagonal entries.
\item The final estimate at step $i$ is identified as the $k$-factor
$\mathbf{T}_k^{\ast}$ in the assumed decomposition
$\mathbf{T} = \mathbf{T}_1^{\ast} \dots \mathbf{T}_n^{\ast}$.
\item One computes now $\mathbf{T}_{k+1}^{\ast}=(\mathbf{T}_1^{\ast} 
\dots \mathbf{T}_k^{\ast})^{-1}\mathbf{T}$ and repeats the procedure.
\item The full algorithm ends when the last estimated matrix in the
decomposition is either an embeddable matrix or a matrix sufficiently
close to the identity. More specifically, when the matrix norm of
the difference between the matrix and identity matrix is at least one
order of magnitude smaller than the matrix norm of the estimated matrix.
Alternatively, when the number of iterations exceeds a pre-fixed 
maximum number of iterations, typically a few thousand, the algorithm 
stops.
\end{enumerate}

We tested $1000$ matrices with principal logarithms having only one negative 
off-diagonal entry and a valid generator was found $945$ times. 
If the number of negative entries is not too large at each step of the
recursive procedure above ($<n^2$) similar results are obtain, which 
indicates an accuracy of around $90$ and $95\%$.
%%%%%%%%%%%%%%%%%%%%%%%%%%%%%%%%%%%%%%%%%%%%%%%%%%%%%%%%%%%%%%%%%%%
\begin{figure}[t]
\centering
\includegraphics[width=0.46\textwidth]{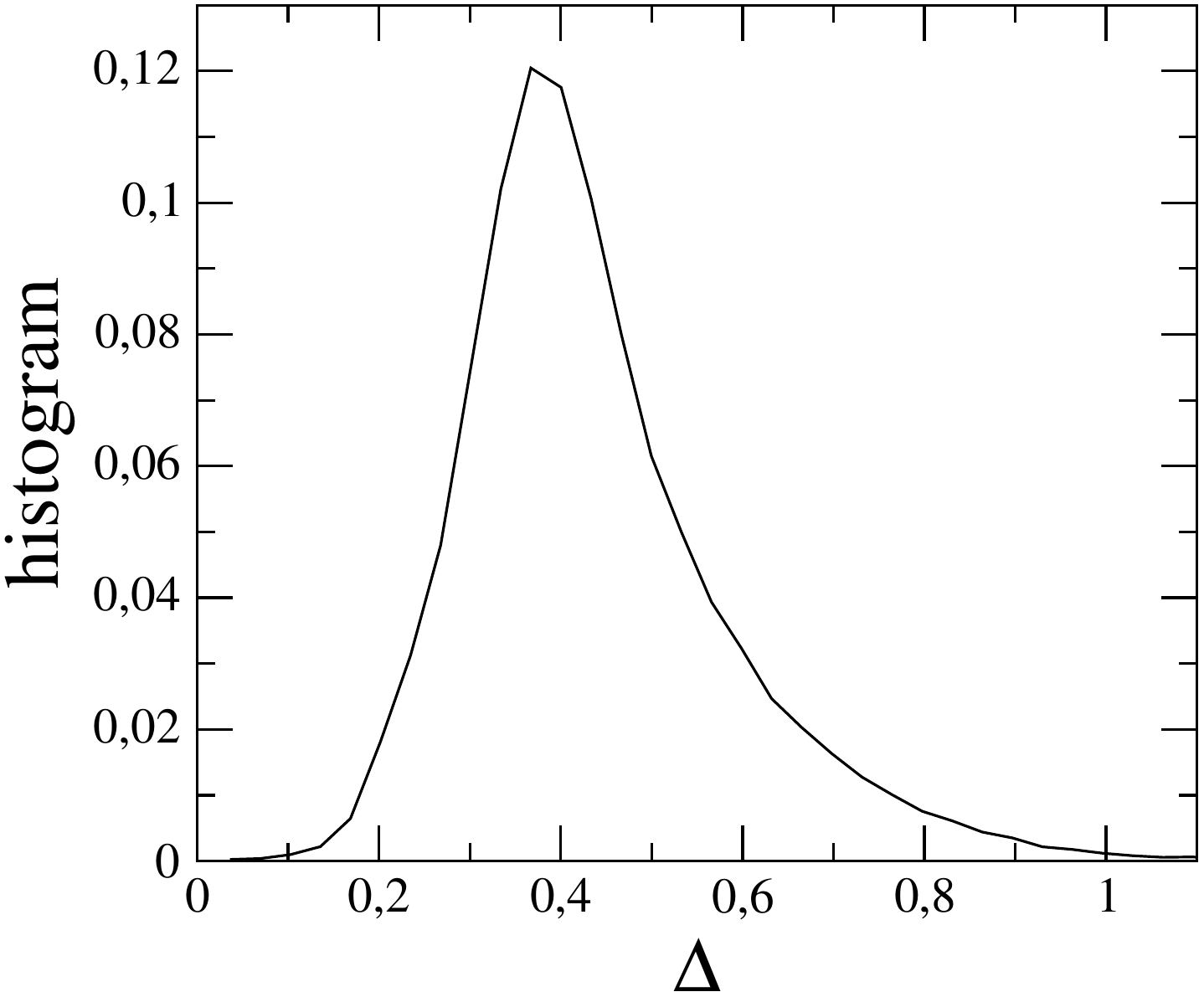}
\caption{\protect
         Histogram of $\Delta$ values, Eq.~(\ref{distance}), 
         from a sample of $43222$ matrices (see text).}
\label{fig02}
\end{figure}
%%%%%%%%%%%%%%%%%%%%%%%%%%%%%%%%%%%%%%%%%%%%%%%%%%%%%%%%%%%%%%%%%%%

In Sec.~\ref{subsec:metrics} we generated matrices with an inhomogeneous 
generator $\mathbf{Q}(t) = \mathbf{Q}_0 + \mathbf{Q}_1(t)$,
integrated them in order to compute a transition matrix $\mathbf{T}(t)$
and produce data series for testing our framework. 
Here, we use the subset of matrices that were correctly detected as
time-inhomogeneous embeddable and estimate one generator as described 
above.
A valid generator was found in around $90\%$ of the generated samples.
%$\mathbf{Q}(t)$ with a probability 
%of about $0.9$, could not produce a valid generator with a 
%probability $0.05$ and was inconclusive with in about $0.05$ of the time.

To evaluate the accuracy of the estimates, we compare 
the modeled transition matrix $\mathbf{T}_{mod}(t,\tau)$ 
with the empirical one, $\mathbf{T}_{emp}(t,\tau)$.
The comparison is based in a normalized distance given by the fraction
of the matrix norm of the difference between both matrices and the
matrix norm of the difference between the modeled matrix and the
identity matrix (initial state):
\begin{equation}
\Delta = \frac{\vert\vert\mathbf{T}_{mod}(t,\tau /2) -
  \mathbf{T}_{emp}(t,\tau /2) \vert\vert_F}{\vert\vert
  \mathbf{T}_{mod}(t,\tau /2) - \mathbf{Id} \vert\vert_F}  ,
\label{distance}
\end{equation}
where $\vert\vert \cdot\vert\vert_F$ is the Frobenius norm.
Figure \ref{fig02} shows an histogram of computed values of the
normalized distance $\Delta$ in Eq.~(\ref{distance}) for all
estimates. Typically the deviations are not larger than $40\%$ of the
deviations from the initial state, where no transition occur. 

Such procedure closes the computational framework for uncovering Markov 
continuous processes from empirical data sets.

%%%%%%%%%%%%%%%%%%%%%%%%%%%%%%%%%%%%%%%%%%%%%%%%%%%%%%%%%%%%%%%%%%%%%%%
\section{Discussion and Conclusions}
\label{sec:conclusions}

We have extend some theoretical results on the inhomogeneous embedding
problem and established a framework which can evaluate empirical data
for detecting the existence of continuous Markov processes.
Eight new proposition were presented and demonstrated concerning the
general case of processes having a time-inhomogeneous generator.
While it was also recently proven that the problem of deriving a
general algorithm capable of solving the embedding problem for any
finite dimension $n$ is NP-hard\cite{NPhard}, 
our implimented algorithm presents acceptable results:
when applied to synthetic data generated from pre-given generators
our framework is able to detect at least $80\%$ of them and, moreover,
returns a good estimate of the generator underlying the data set.
Thus, our algorithm enables one to find a time-inhomogeneous generator of
any transition matrix with a real-valued logarithm.

Concerning the new proposition demonstrated above for inhomogeneous
transition matrices, there are e.g.~some extensions of the $LU$
decomposition theorem, Prop.~\ref{LU}, that can be interesting for future work. 
Namely, the quasi LU decomposition\cite{bueno2007minimum}, the ULU
decomposition\cite{toffoli1997almost} and the 
LULU factorization\cite{strang1997every}. 

This framework is now able to be straightforwardly applied to specific
sets of data for evaluating hidden continuous Markov processes. 
Indeed, since the transition matrix defines a
specific Markov chain, our framework can be taken as a possibility
for accessing continuous hidden processes in (time-dependent)
Markov chains found in e.g.~models for polymer growth processes or
enzyme activity.

For specific applications, our framework can be used for three types
of stochastic data sets: (i) one where only the initial and final
configuration of the system is given; (ii) one where all possible
state transitions are defined through a probability value between the
start and end of the observation period and (iii) the transition
between the beginning of intermediate instants till the end of 
the observation. In this paper we dealt typically with type (ii) data
sets, while in previous works\cite{lencastre1,lencastre2} we
considered mainly type (iii).  
Type (i) is typically not well-defined and additional cautions must be
taken. 

One important interdisciplinary application is, of course, in
economics and finance, when addressing rating matrices:
if ratings do indeed reflect a natural (continuous) economic
process, the extracted rating matrices must have a proper 
generator\cite{Metz2007} .
This problem as already addressed by us\cite{lencastre1,lencastre2}
in the particular case of homogeneous transition matrices derived by 
rating agencies.
Further, our methodology could be extended to other situations where
correlation matrices are taken for describing the macroscopic
state of a financial system\cite{yuri}. 
With a proper normalization such correlation matrices can be taken,
in an algebraic sense, as transition matrices and therefore the
framework described above is applicable.

%%%%%%%%%%
\section*{Acknowledgments}

The authors thank 
{\it Deutscher Akademischer Austauschdienst} (DAAD) and 
{\it Funda\c{c}\~ao para a Ci\^encia e a Tecnologia} (FCT) for
support from bilateral collaboration DRI/DAAD/1208/2013.
FR thanks FCT for the fellowship SFRH/BPD/65427/2009. 
TR acknowledges support from the Royal Society.
PGL thanks German Environment Ministry for financial support.

%%%%%%%%%%%%%%%%%%%%%%%%%%%%%%%%%%%%%%%%%%%%%%%%%%%%%%%%%%%%%%%%%%%%%%%%%
%%%%%%%%%%%   APPENDICES   %%%%%%%%%%%%%%%%%%%%%%%%%%%%%%%%%%%%%%%%%%%%%%
%%%%%%%%%%%%%%%%%%%%%%%%%%%%%%%%%%%%%%%%%%%%%%%%%%%%%%%%%%%%%%%%%%%%%%%%%
\appendix
%%%%%%%%%%%%%%%%%%%%%%%%%%%%%%%%%%%%%%%%%%%%%%%%%%%%%%%%%%%%%%%%%%%%%%%%%
\section{Additional results on the time-inhomogeneous embeddable problem}
\label{append:results}

Here we present additional results concerning the existence
of inhomogeneous generators. These results serve for proving the 
theorems implemented above and provide theoretical consistency
to our framework.
The first result is a sufficient condition concerning a possible
decomposition of transition matrices:

\begin{proposition}
If $\mathbf{T}$ is an $n$-dimensional triangular transition matrix, 
then it has an inhomogeneous generator, which can be defined from a 
decomposition of the transition matrix as $\mathbf{T}=
e^{\mathbf{Q}_1}\cdots e^{\mathbf{Q}_{n-1}}$
where $\mathbf{Q}_i$ are time-homogeneous generators of some elementary 
transition matrix.
\label{prop:append1}
\end{proposition}

\begin{proof}
The proof is given by induction. 
For $n=2$, the triangular transition matrix $\mathbf{T}$ can be parameterised 
by one single parameter $p\in[0,1]$:
\begin{equation}
\mathbf{T}=\begin{pmatrix}1-p & p\\0&1\end{pmatrix}\, .
\end{equation}
It is straightforward to see that $\mathbf{T}=e^{\mathbf{Q}}$ with
\begin{equation}
\mathbf{Q}=\begin{pmatrix}\log(1-p)&-\log(1-p)\\0&0\end{pmatrix}\, .
\end{equation}
Since $\log(1-p)<0$, $\mathbf{Q}$ is indeed a generator matrix. 

We now consider an triangular transition matrix of arbitrary dimension
$n$ and treat the rightmost column separately, yielding
\begin{equation}
\mathbf{T}=\begin{pmatrix}\mathbf{A}&\underline{a}\\\ 
\underline{0}^{\top}&1\end{pmatrix}\, ,
\end{equation}
where $\mathbf{A}$ is an $(n-1)\times(n-1)$ triangular matrix, $\underline{a}$ is 
a column-vector with $n-1$ non-negative entries and $\underline{0}^{\top}$ is a 
row-vector of $n-1$ zeros. 
Since $\mathbf{T}$ is a transition matrix, for all $i=1,\dots,n-1$ one has
\begin{equation}
\sum_j A_{ij}=1-a_i\, .
\end{equation}
Introducing a $(n-1)$-dimensional triangular transition matrix 
$\mathbf{T}^\prime$ with entries $T^\prime_{ij}=\frac{A_{ij}}{1-a_i}$,
one reads
\begin{equation}
\mathbf{T}=\begin{pmatrix}\mathbf{I}-\text{diag}(\underline{a})
&\underline{a}\\\underline{0}^{\top}&1\end{pmatrix}\begin{pmatrix}\mathbf{T}^\prime&\underline{0}\\\ \underline{0}^{\top}&1\end{pmatrix}\, ,
\end{equation}
where $\text{diag}(\underline{a})$ is the $(n-1)$-dimensional diagonal
matrix with entries taken from vector $\underline{a}$.  
The first matrix above is embeddable since
\begin{equation}
\begin{pmatrix}
\mathbf{I}-\text{diag}(\underline{a}) & \underline{a}\\
\underline{0}^{\top}                  & 1
\end{pmatrix}
=\exp{
\begin{pmatrix}
\text{diag}(\log(\underline{1}-\underline{a})) &
-\log(\underline{1}-\underline{a}) \\
\underline{0}^{\top} & 0
\end{pmatrix}  
}
\end{equation}
and the second matrix can be further decomposed as
\begin{equation}
\mathbf{T} =
\begin{pmatrix}
\mathbf{I}-\text{diag}(\underline{a}) & \underline{a}\\
\underline{0}^{\top}                  & 1
\end{pmatrix}
\left (
\begin{array}{ccc}
\mathbf{I}-\text{diag}(\underline{b}) & \underline{b} & 0\\
\underline{0}^{\top}                  & 1             & 0\\
0                                     & 0             & 1
\end{array}
\right )
\left (
\begin{array}{ccc}
\mathbf{T}^\prime & \underline{0} & \underline{0}\\
0 & 1 & 0 \\
\underline{0}^{\top} & 0 & 1 
\end{array}
\right )
\, .
\end{equation}

Therefore, we arrive to a decomposition of the form 
$\mathbf{T}=e^{\mathbf{Q^\prime_1}}\cdots e^{\mathbf{Q^\prime_{n-1}}}$ 
for generator matrices 
$\mathbf{Q^\prime_1},\ldots,\mathbf{Q^\prime_{n-1}}$ with
\begin{equation}
\mathbf{Q_{1}}=\begin{pmatrix}\text{diag}(\log(\underline{1}-\underline{a}))&-\log(\underline{1}-\underline{a})\\\underline{0}^{\top}&0\end{pmatrix}
\end{equation}
and
\begin{equation}
\mathbf{Q_k} =\begin{pmatrix} \mathbf{Q^\prime_{k-1}}&\underline{0}\\\underline{0}^{\top}&0\end{pmatrix}
\end{equation}
for $k=2,\ldots,n-1$.
\end{proof}

One could implement Prop.~\ref{prop:append1} by finding a product
of $n$-dimensional square matrices $\prod_i \mathbf{A}^{(i)}$ 
where each matrix $\mathbf{A}^{(i)}$ has only one off-diagonal
non-zero element and if for matrix $\mathbf{A}^{(k)}$ one has $A^{(k)}_{ij}
\neq 0 $, then for all other matrices $\mathbf{A}^{(l)}$ ($l\neq k$) 
one has $A^{(l)}_{ij} = 0$.
If that product has $m = n(n-1)$ terms, we can solve 
$\prod_i \mathbf{A}^{(i)} = \mathbf{T}$ as a linear system of equations 
with $n$ equations and $n$ unknowns.
Having this, we define the following distance for the 
$\mathbf{A}$-factorization:
\begin{equation}
d_{S_2}= \min_k \{ \min_{i,j} \{ \tfrac{A^{k}_{ij}}{\sigma_{A^{k}_{ij}}} \}  \} \, ,
\label{distanceS2}
\end{equation} 
where $\sigma_{A^{n}_{ij}}$ is the dispersion associated with the
entry $A^{n}_{ij}$. 
If $d_{S_2}>2$, we statistically infer that a generator exists.
Notice that, it is possible to prove that the $LU$ decomposition is a 
particular case of the factorization in Eq.~(\ref{distanceS2}). 
%The decomposition in $\mathbf{A}$ matrices, however, is much harder to implement computationally.

One additional necessary condition that may be useful in some cases is the 
following one:
\begin{proposition}
An irreducible matrix $\mathbf{T}$,
i.e.~it cannot be placed into block upper-triangular form by
simultaneous row or column permutations,
is time-inhomogeneous embeddable only 
if, for at least in one row there is more than one non-zero
off-diagonal entry.
\label{necessary3}
\end{proposition}
\begin{proof}
If $\mathbf{T}$ is time-inhomogeneous embeddable, then from
Prop.~\ref{prop:append1}, $\mathbf{T}$ can be written as a product $n$ of
embeddable matrices $\mathbf{P}^{(k)}=\exp{\mathbf{Q}^{(k)}}$. Assume,
without loss of generality that all matrices $\mathbf{P}^{(k)}$ are
time-homogeneous embeddable. 

Since no matrix $\mathbf{P}^{(i)}$ has no zeros in the diagonal
entries, from Props.~\ref{necessary1} and \ref{necessary2}, the
product of an irreducible matrix by an embeddable matrix is always
irreducible. 
Notice that if any of the matrices $\mathbf{P}^{(k)}$ is 
time-homogeneous embeddable, then from Prop.~\ref{ncondition1}c),
$\mathbf{T}$ will have no zero entries. 

Let us consider $\mathbf{P}^{(k)}$ such that the product
$\mathbf{P}^{(1)} \dots \mathbf{P}^{(k)}$ is irreducible but
$\mathbf{P}^{(1)} \dots \mathbf{P}^{(k-1)}$ is not. 
Since we assume, without loss of generality that $\mathbf{P}^{(k)}$ is
not the identity matrix, $P^{(k)}_{ij}>0$ for at least one $j\neq i$.  
Then, for $m\leq k$ there is one $l$ for which $P_{li}^{(m)}>0$.  
Thus $T_{ij}> 0 $ and  $T_{lj}> 0 $. 
\end{proof}

Proposition \ref{necessary3} is not a condition we can evaluate 
for empirical systems. 
Nonetheless it might be useful if one has some apriori knowledge 
about the dynamics of the system. 

Another sufficient condition for time-inhomogeneous generators concerns
situations when the matrices have non-negative entries:
\begin{proposition}
Totally non-negative transition matrices, i.e.~matrices
$\mathbf{T}(t)$ for which all submatrices have positive determinant,
have an inhomogeneous generator $\mathbf{Q}(t)$. 
\label{totally-non-negative}
\end{proposition}
\begin{proof}
It was proved \cite{cryer1973lu} that the LU factorization of any
totally non-negative matrix is composed by a totally non-negative
lower diagonal matrix $\mathbf{L}$ and a totally non-negative upper
diagonal  $\mathbf{U}$ . If a matrix is totally non-negative, then it
has only non-negative elements, thus in particular $\mathbf{L}$ and
$\mathbf{U}$ are matrices with non-negative elements.  
\end{proof}
%%%%%%%%%%%%%%%%%%%%%%%%%%%%%%%%%%%%%%%%%%%%%%%%%%%%%%%%%%%%%%%%%%%%
%\begin{thebibliography}{99}
\bibliographystyle{apsrev4-1}
\bibliography{MatrixBib.bib}

\begin{thebibliography}{26}%
\makeatletter
\providecommand \@ifxundefined [1]{%
 \@ifx{#1\undefined}
}%
\providecommand \@ifnum [1]{%
 \ifnum #1\expandafter \@firstoftwo
 \else \expandafter \@secondoftwo
 \fi
}%
\providecommand \@ifx [1]{%
 \ifx #1\expandafter \@firstoftwo
 \else \expandafter \@secondoftwo
 \fi
}%
\providecommand \natexlab [1]{#1}%
\providecommand \enquote  [1]{``#1''}%
\providecommand \bibnamefont  [1]{#1}%
\providecommand \bibfnamefont [1]{#1}%
\providecommand \citenamefont [1]{#1}%
\providecommand \href@noop [0]{\@secondoftwo}%
\providecommand \href [0]{\begingroup \@sanitize@url \@href}%
\providecommand \@href[1]{\@@startlink{#1}\@@href}%
\providecommand \@@href[1]{\endgroup#1\@@endlink}%
\providecommand \@sanitize@url [0]{\catcode `\\12\catcode `\$12\catcode
  `\&12\catcode `\#12\catcode `\^12\catcode `\_12\catcode `\%12\relax}%
\providecommand \@@startlink[1]{}%
\providecommand \@@endlink[0]{}%
\providecommand \url  [0]{\begingroup\@sanitize@url \@url }%
\providecommand \@url [1]{\endgroup\@href {#1}{\urlprefix }}%
\providecommand \urlprefix  [0]{URL }%
\providecommand \Eprint [0]{\href }%
\providecommand \doibase [0]{http://dx.doi.org/}%
\providecommand \selectlanguage [0]{\@gobble}%
\providecommand \bibinfo  [0]{\@secondoftwo}%
\providecommand \bibfield  [0]{\@secondoftwo}%
\providecommand \translation [1]{[#1]}%
\providecommand \BibitemOpen [0]{}%
\providecommand \bibitemStop [0]{}%
\providecommand \bibitemNoStop [0]{.\EOS\space}%
\providecommand \EOS [0]{\spacefactor3000\relax}%
\providecommand \BibitemShut  [1]{\csname bibitem#1\endcsname}%
\let\auto@bib@innerbib\@empty
%</preamble>
\bibitem [{\citenamefont {Elfving}(1937)}]{Elfving1937}%
  \BibitemOpen
  \bibfield  {author} {\bibinfo {author} {\bibfnamefont {G.}~\bibnamefont
  {Elfving}},\ }\href@noop {} {\emph {\bibinfo {title} {Zur theorie der
  Markoffschen ketten}}},\ Acta Societatis scientiarum Fennicae, Nova series A,
  T.2, 8\ (\bibinfo  {publisher} {Harrassowitz},\ \bibinfo {address}
  {Leipzig},\ \bibinfo {year} {1937})\BibitemShut {NoStop}%
\bibitem [{\citenamefont {Kingman}(1962)}]{Kingman1962}%
  \BibitemOpen
  \bibfield  {author} {\bibinfo {author} {\bibfnamefont {J.}~\bibnamefont
  {Kingman}},\ }\href@noop {} {\bibfield  {journal} {\bibinfo  {journal}
  {Probability Theory and Related Fields}\ }\textbf {\bibinfo {volume} {1}},\
  \bibinfo {pages} {14} (\bibinfo {year} {1962})}\BibitemShut {NoStop}%
\bibitem [{\citenamefont {Esposito}\ and\ \citenamefont {Van~den
  Broeck}(2010)}]{esposito2010three}%
  \BibitemOpen
  \bibfield  {author} {\bibinfo {author} {\bibfnamefont {M.}~\bibnamefont
  {Esposito}}\ and\ \bibinfo {author} {\bibfnamefont {C.}~\bibnamefont {Van~den
  Broeck}},\ }\href@noop {} {\bibfield  {journal} {\bibinfo  {journal}
  {Physical Review E}\ }\textbf {\bibinfo {volume} {82}},\ \bibinfo {pages}
  {011143} (\bibinfo {year} {2010})}\BibitemShut {NoStop}%
\bibitem [{\citenamefont {Privman}(2005)}]{privman2005nonequilibrium}%
  \BibitemOpen
  \bibfield  {author} {\bibinfo {author} {\bibfnamefont {V.}~\bibnamefont
  {Privman}},\ }\href@noop {} {\emph {\bibinfo {title} {Nonequilibrium
  statistical mechanics in one dimension}}}\ (\bibinfo  {publisher} {Cambridge
  University Press},\ \bibinfo {year} {2005})\BibitemShut {NoStop}%
\bibitem [{\citenamefont {Shintani}\ and\ \citenamefont
  {Shinomoto}(2012)}]{shintani2012}%
  \BibitemOpen
  \bibfield  {author} {\bibinfo {author} {\bibfnamefont {T.}~\bibnamefont
  {Shintani}}\ and\ \bibinfo {author} {\bibfnamefont {S.}~\bibnamefont
  {Shinomoto}},\ }\href@noop {} {\bibfield  {journal} {\bibinfo  {journal}
  {Phys. Rev. E}\ }\textbf {\bibinfo {volume} {85}},\ \bibinfo {pages} {041139}
  (\bibinfo {year} {2012})}\BibitemShut {NoStop}%
\bibitem [{\citenamefont {Lencastre}\ \emph
  {et~al.}(2015{\natexlab{a}})\citenamefont {Lencastre}, \citenamefont
  {Raischel}, \citenamefont {Lind},\ and\ \citenamefont {Rogers}}]{lencastre1}%
  \BibitemOpen
  \bibfield  {author} {\bibinfo {author} {\bibfnamefont {P.}~\bibnamefont
  {Lencastre}}, \bibinfo {author} {\bibfnamefont {F.}~\bibnamefont {Raischel}},
  \bibinfo {author} {\bibfnamefont {P.~G.}\ \bibnamefont {Lind}}, \ and\
  \bibinfo {author} {\bibfnamefont {T.}~\bibnamefont {Rogers}},\ }in\
  \href@noop {} {\emph {\bibinfo {booktitle} {3rd SMDTA Conference
  Proceedings}}}\ (\bibinfo {year} {2015})\ pp.\ \bibinfo {pages}
  {405--414}\BibitemShut {NoStop}%
\bibitem [{\citenamefont {Lencastre}\ \emph
  {et~al.}(2015{\natexlab{b}})\citenamefont {Lencastre}, \citenamefont
  {Raischel},\ and\ \citenamefont {Lind}}]{lencastre2}%
  \BibitemOpen
  \bibfield  {author} {\bibinfo {author} {\bibfnamefont {P.}~\bibnamefont
  {Lencastre}}, \bibinfo {author} {\bibfnamefont {F.}~\bibnamefont {Raischel}},
  \ and\ \bibinfo {author} {\bibfnamefont {P.~G.}\ \bibnamefont {Lind}},\
  }\href@noop {} {\bibfield  {journal} {\bibinfo  {journal} {Journal of
  Physics: Conf. Ser.}\ }\textbf {\bibinfo {volume} {574}},\ \bibinfo {pages}
  {012151} (\bibinfo {year} {2015}{\natexlab{b}})}\BibitemShut {NoStop}%
\bibitem [{\citenamefont {Chen}\ and\ \citenamefont
  {Chen}(2011)}]{3x3generator}%
  \BibitemOpen
  \bibfield  {author} {\bibinfo {author} {\bibfnamefont {Y.}~\bibnamefont
  {Chen}}\ and\ \bibinfo {author} {\bibfnamefont {J.}~\bibnamefont {Chen}},\
  }\href@noop {} {\bibfield  {journal} {\bibinfo  {journal} {Journal of
  Theoretical Probability}\ }\textbf {\bibinfo {volume} {24}},\ \bibinfo
  {pages} {928} (\bibinfo {year} {2011})}\BibitemShut {NoStop}%
\bibitem [{\citenamefont {Davies}(2010)}]{Davies2010}%
  \BibitemOpen
  \bibfield  {author} {\bibinfo {author} {\bibfnamefont {E.}~\bibnamefont
  {Davies}},\ }\href@noop {} {\bibfield  {journal} {\bibinfo  {journal}
  {Electronic Journal of Probability}\ }\textbf {\bibinfo {volume} {15}},\
  \bibinfo {pages} {1474} (\bibinfo {year} {2010})}\BibitemShut {NoStop}%
\bibitem [{\citenamefont {Israel}\ \emph {et~al.}(2001)\citenamefont {Israel},
  \citenamefont {Rosenthal},\ and\ \citenamefont {Wei}}]{Israel2001}%
  \BibitemOpen
  \bibfield  {author} {\bibinfo {author} {\bibfnamefont {R.}~\bibnamefont
  {Israel}}, \bibinfo {author} {\bibfnamefont {J.}~\bibnamefont {Rosenthal}}, \
  and\ \bibinfo {author} {\bibfnamefont {J.}~\bibnamefont {Wei}},\ }\href@noop
  {} {\bibfield  {journal} {\bibinfo  {journal} {Mathematical Finance}\
  }\textbf {\bibinfo {volume} {11}},\ \bibinfo {pages} {245} (\bibinfo {year}
  {2001})}\BibitemShut {NoStop}%
\bibitem [{\citenamefont {Karpelevich}(1951)}]{karpelevich1951characteristic}%
  \BibitemOpen
  \bibfield  {author} {\bibinfo {author} {\bibfnamefont {F.~I.}\ \bibnamefont
  {Karpelevich}},\ }\href@noop {} {\bibfield  {journal} {\bibinfo  {journal}
  {Izvestiya Rossiiskoi Akademii Nauk. Seriya Matematicheskaya}\ }\textbf
  {\bibinfo {volume} {15}},\ \bibinfo {pages} {361} (\bibinfo {year}
  {1951})}\BibitemShut {NoStop}%
\bibitem [{\citenamefont {Higham}\ and\ \citenamefont
  {Lin}(2011)}]{Higham2011}%
  \BibitemOpen
  \bibfield  {author} {\bibinfo {author} {\bibfnamefont {N.~J.}\ \bibnamefont
  {Higham}}\ and\ \bibinfo {author} {\bibfnamefont {L.}~\bibnamefont {Lin}},\
  }\href@noop {} {\bibfield  {journal} {\bibinfo  {journal} {L Algebra and its
  Applications}\ }\textbf {\bibinfo {volume} {435}},\ \bibinfo {pages} {448}
  (\bibinfo {year} {2011})}\BibitemShut {NoStop}%
\bibitem [{\citenamefont {Runnenberg}(1962)}]{runnenburg1962elfving}%
  \BibitemOpen
  \bibfield  {author} {\bibinfo {author} {\bibfnamefont {J.}~\bibnamefont
  {Runnenberg}},\ }\href@noop {} {\bibfield  {journal} {\bibinfo  {journal}
  {Proceedings of the KNAW}\ }\textbf {\bibinfo {volume} {65}},\ \bibinfo
  {pages} {536} (\bibinfo {year} {1962})}\BibitemShut {NoStop}%
\bibitem [{\citenamefont {Gray}(2005)}]{gray2005toeplitz}%
  \BibitemOpen
  \bibfield  {author} {\bibinfo {author} {\bibfnamefont {R.~M.}\ \bibnamefont
  {Gray}},\ }\href@noop {} {\bibfield  {journal} {\bibinfo  {journal}
  {Communications and Information Theory}\ }\textbf {\bibinfo {volume} {2}},\
  \bibinfo {pages} {155} (\bibinfo {year} {2005})}\BibitemShut {NoStop}%
\bibitem [{\citenamefont {Baake}\ and\ \citenamefont
  {Schl{\"a}gel}(2011)}]{baake2011peano}%
  \BibitemOpen
  \bibfield  {author} {\bibinfo {author} {\bibfnamefont {M.}~\bibnamefont
  {Baake}}\ and\ \bibinfo {author} {\bibfnamefont {U.}~\bibnamefont
  {Schl{\"a}gel}},\ }\href@noop {} {\bibfield  {journal} {\bibinfo  {journal}
  {Proceedings of the Steklov Institute of Mathematics}\ }\textbf {\bibinfo
  {volume} {275}},\ \bibinfo {pages} {155} (\bibinfo {year}
  {2011})}\BibitemShut {NoStop}%
\bibitem [{\citenamefont {Pachpatte}(1998)}]{pachpatte98}%
  \BibitemOpen
  \bibfield  {author} {\bibinfo {author} {\bibfnamefont {B.}~\bibnamefont
  {Pachpatte}},\ }\href@noop {} {\emph {\bibinfo {title} {Inequalities for
  differential and integral equations}}}\ (\bibinfo  {publisher} {San Diego,
  Academic Press},\ \bibinfo {year} {1998})\BibitemShut {NoStop}%
\bibitem [{\citenamefont {Higham}(2008)}]{higham2008functions}%
  \BibitemOpen
  \bibfield  {author} {\bibinfo {author} {\bibfnamefont {N.~J.}\ \bibnamefont
  {Higham}},\ }\href@noop {} {\emph {\bibinfo {title} {Functions of matrices:
  theory and computation}}}\ (\bibinfo  {publisher} {Siam},\ \bibinfo {year}
  {2008})\BibitemShut {NoStop}%
\bibitem [{\citenamefont {Sherlaw-Johnson}\ \emph {et~al.}(1995)\citenamefont
  {Sherlaw-Johnson}, \citenamefont {Gallivan},\ and\ \citenamefont
  {Burridge}}]{Burrige1995}%
  \BibitemOpen
  \bibfield  {author} {\bibinfo {author} {\bibfnamefont {C.}~\bibnamefont
  {Sherlaw-Johnson}}, \bibinfo {author} {\bibfnamefont {S.}~\bibnamefont
  {Gallivan}}, \ and\ \bibinfo {author} {\bibfnamefont {J.}~\bibnamefont
  {Burridge}},\ }\href@noop {} {\bibfield  {journal} {\bibinfo  {journal}
  {Journal of the Operational Research Society}\ ,\ \bibinfo {pages} {405}}
  (\bibinfo {year} {1995})}\BibitemShut {NoStop}%
\bibitem [{\citenamefont {Jafry}\ and\ \citenamefont
  {Schuermann}(2004)}]{jafry2004measurement}%
  \BibitemOpen
  \bibfield  {author} {\bibinfo {author} {\bibfnamefont {Y.}~\bibnamefont
  {Jafry}}\ and\ \bibinfo {author} {\bibfnamefont {T.}~\bibnamefont
  {Schuermann}},\ }\href@noop {} {\bibfield  {journal} {\bibinfo  {journal}
  {Journal of Banking \& Finance}\ }\textbf {\bibinfo {volume} {28}},\ \bibinfo
  {pages} {2603} (\bibinfo {year} {2004})}\BibitemShut {NoStop}%
\bibitem [{\citenamefont {Cubitt}\ \emph {et~al.}(2012)\citenamefont {Cubitt},
  \citenamefont {Eisert},\ and\ \citenamefont {Wolf}}]{NPhard}%
  \BibitemOpen
  \bibfield  {author} {\bibinfo {author} {\bibfnamefont {T.~S.}\ \bibnamefont
  {Cubitt}}, \bibinfo {author} {\bibfnamefont {J.}~\bibnamefont {Eisert}}, \
  and\ \bibinfo {author} {\bibfnamefont {M.~M.}\ \bibnamefont {Wolf}},\
  }\href@noop {} {\bibfield  {journal} {\bibinfo  {journal} {Physical review
  letters}\ }\textbf {\bibinfo {volume} {108}},\ \bibinfo {pages} {120503}
  (\bibinfo {year} {2012})}\BibitemShut {NoStop}%
\bibitem [{\citenamefont {Bueno}\ and\ \citenamefont
  {Johnson}(2007)}]{bueno2007minimum}%
  \BibitemOpen
  \bibfield  {author} {\bibinfo {author} {\bibfnamefont {M.}~\bibnamefont
  {Bueno}}\ and\ \bibinfo {author} {\bibfnamefont {C.}~\bibnamefont
  {Johnson}},\ }\href@noop {} {\bibfield  {journal} {\bibinfo  {journal}
  {Linear Algebra and its Applications}\ }\textbf {\bibinfo {volume} {427}},\
  \bibinfo {pages} {99} (\bibinfo {year} {2007})}\BibitemShut {NoStop}%
\bibitem [{\citenamefont {Toffoli}(1997)}]{toffoli1997almost}%
  \BibitemOpen
  \bibfield  {author} {\bibinfo {author} {\bibfnamefont {T.}~\bibnamefont
  {Toffoli}},\ }\href@noop {} {\bibfield  {journal} {\bibinfo  {journal}
  {Linear algebra and its applications}\ }\textbf {\bibinfo {volume} {259}},\
  \bibinfo {pages} {31} (\bibinfo {year} {1997})}\BibitemShut {NoStop}%
\bibitem [{\citenamefont {Strang}(1997)}]{strang1997every}%
  \BibitemOpen
  \bibfield  {author} {\bibinfo {author} {\bibfnamefont {G.}~\bibnamefont
  {Strang}},\ }\href@noop {} {\bibfield  {journal} {\bibinfo  {journal} {Linear
  algebra and its applications}\ }\textbf {\bibinfo {volume} {265}},\ \bibinfo
  {pages} {165} (\bibinfo {year} {1997})}\BibitemShut {NoStop}%
\bibitem [{\citenamefont {Metz}\ and\ \citenamefont {Cantor}(2007)}]{Metz2007}%
  \BibitemOpen
  \bibfield  {author} {\bibinfo {author} {\bibfnamefont {A.}~\bibnamefont
  {Metz}}\ and\ \bibinfo {author} {\bibfnamefont {R.}~\bibnamefont {Cantor}},\
  }\href@noop {} {\emph {\bibinfo {title} {Introducing Moody's Credit
  Transition Model}}}\ (\bibinfo  {publisher} {Moody's Analytics},\ \bibinfo
  {year} {2007})\BibitemShut {NoStop}%
\bibitem [{\citenamefont {M\"unnix}\ \emph {et~al.}(2012)\citenamefont
  {M\"unnix}, \citenamefont {Shimada}, \citenamefont {Sch\"afer}, \citenamefont
  {Leyvraz}, \citenamefont {Seligman}, \citenamefont {Guhr},\ and\
  \citenamefont {Stanley}}]{yuri}%
  \BibitemOpen
  \bibfield  {author} {\bibinfo {author} {\bibfnamefont {M.}~\bibnamefont
  {M\"unnix}}, \bibinfo {author} {\bibfnamefont {T.}~\bibnamefont {Shimada}},
  \bibinfo {author} {\bibfnamefont {R.}~\bibnamefont {Sch\"afer}}, \bibinfo
  {author} {\bibfnamefont {F.}~\bibnamefont {Leyvraz}}, \bibinfo {author}
  {\bibfnamefont {T.}~\bibnamefont {Seligman}}, \bibinfo {author}
  {\bibfnamefont {T.}~\bibnamefont {Guhr}}, \ and\ \bibinfo {author}
  {\bibfnamefont {H.}~\bibnamefont {Stanley}},\ }\href@noop {} {\bibfield
  {journal} {\bibinfo  {journal} {Scientific Reports}\ }\textbf {\bibinfo
  {volume} {2}},\ \bibinfo {pages} {644} (\bibinfo {year} {2012})}\BibitemShut
  {NoStop}%
\bibitem [{\citenamefont {Cryer}(1973)}]{cryer1973lu}%
  \BibitemOpen
  \bibfield  {author} {\bibinfo {author} {\bibfnamefont {C.~W.}\ \bibnamefont
  {Cryer}},\ }\href@noop {} {\bibfield  {journal} {\bibinfo  {journal} {Linear
  Algebra and its Applications}\ }\textbf {\bibinfo {volume} {7}},\ \bibinfo
  {pages} {83} (\bibinfo {year} {1973})}\BibitemShut {NoStop}%
\end{thebibliography}%
%\end{thebibliography}

%%%%%%%%%%%%%%%%%%%%%%%%%%%%%%%%%%%%%%%%%%%%%%%%%%%%%%%%%%%%%%%%%%%%
%%%%%%%%%%%%%%%%%%%%%%%%%%%%%%%%%%%%%%%%%%%%%%%%%%%%%%%%%%%%%%%%%%%%
%%%%%%%%%%%%%%%%%%%%%%%%%%%%%%%%%%%%%%%%%%%%%%%%%%%%%%%%%%%%%%%%%%%%
\end{document}